\documentclass[oneside,11pt]{article}

\usepackage{subcaption}
\def\ISDRAFT{0}

\usepackage{subcaption}
\def\ISDRAFT{0}

\newcommand{\WATERMARKTEXT}{}

\usepackage[utf8]{inputenc}
\usepackage{ifthen,cite,mathrsfs, amsmath, amssymb, amsthm, braket, bm, graphicx, mathtools,enumitem,color}
\usepackage[english]{babel}
\usepackage[colorlinks=true, allcolors=blue]{hyperref}
\hypersetup{final}
\usepackage[hang,flushmargin]{footmisc}

\newcommand{\numorbitals}{N_\mathrm{b}}
\newcommand{\ctTBprefactor}[1]{
	\ifthenelse{\equal{#1}{}}
	{\bm{h}}
	{{h}_{#1}}
}
\newcommand{\ctTBexponent}[1]{
	\ifthenelse{\equal{#1}{}}
		{\bm{\gamma}}
		{\gamma_{#1}}
}
\newcommand{\ctHamregularity}{\nu}
\newcommand{\ctnoninterpen}{\mathfrak{m}}
\newcommand{\ctCT}{\gamma_\mathrm{CT}}
\newcommand{\ctGamma}{{\Upsilon}}
\newcommand{\ctgap}{\mathsf{g}}
\newcommand{\ctgapHom}{\mathsf{g}^\mathrm{ref}}
\newcommand{\ctgapfinitetemp}{\mathsf{d}}
\newcommand{\ctgapfinitetempHom}{\mathsf{d}^\mathrm{ref}}
\newcommand{\ctAwayFromSingularity}{\mathsf{b}}
\newcommand{\Ham}{\mathcal{H}}
\newcommand{\W}{{\dot{\mathscr W}}^{1,2}}
\let\originalleft\left
\let\originalright\right
\renewcommand{\left}{\mathopen{}\mathclose\bgroup\originalleft}
\renewcommand{\right}{\aftergroup\egroup\originalright}
\renewcommand{\Re}[1]{\operatorname{Re}\left(#1\right)}
\renewcommand{\Im}[1]{\operatorname{Im}\left(#1\right)}
\renewcommand{\leq}{\leqslant}\renewcommand{\geq}{\geqslant}
\renewcommand{\above}[2]{\genfrac{}{}{0pt}{}{#1}{#2}}

\newtheorem*{fact*}{Fact}
\newtheorem{theorem}{Theorem}
\newtheorem*{theorem*}{Theorem}

\newtheorem*{assumption*}{Assumption}

\newtheorem{lemma}[theorem]{Lemma}
\newtheorem{prop}[theorem]{Proposition}

\theoremstyle{remark}
\newtheorem{remark}{Remark}

\newcommand{\thistheoremname}{}
\newtheorem*{genericthm*}{\textup{\thistheoremname}}
\newenvironment{assumption}[1]{
	\renewcommand{\thistheoremname}{#1}%
	\begin{genericthm*}}
	{\end{genericthm*}}
\numberwithin{equation}{section}

\usepackage{cleveref}
\crefformat{equation}{\textup{(#2#1#3)}}
\crefformat{section}{\S#2#1#3} 
\crefformat{subsection}{\S#2#1#3}
\crefformat{subsubsection}{\S#2#1#3}
\crefformat{thm}{Theorem~#2#1#3}

\if\ISDRAFT1
	\usepackage{showkeys}
	\renewcommand*\showkeyslabelformat[1]{%
		\fbox{\parbox[t]{0.9\marginparwidth}{\raggedright\normalfont\tiny\ttfamily\path{#1}}}}
	\usepackage[textwidth=36mm]{todonotes}
\fi
\usepackage{geometry}
	\usepackage[all]{background}
	\SetBgContents{\WATERMARKTEXT}
	\SetBgPosition{5cm,-3.5cm}
	\SetBgOpacity{0.15}
	\SetBgAngle{-90.0}
	\SetBgScale{3.5}
	\SetBgColor{black!40}

\title{Locality of Interatomic Forces in \\ Tight Binding Models for Insulators}
\date{\today}
\author{Christoph Ortner and Jack Thomas and Huajie Chen}

\begin{document}

\maketitle

\begin{abstract} 
The tight binding model is a minimalistic electronic structure model for predicting properties of materials and molecules. For insulators at zero Fermi-temperature we show that the potential energy surface of this model can be decomposed into exponentially localised site energy contributions, thus providing qualitatively sharp estimates on the interatomic interaction range which justifies a range of multi-scale models. For insulators at finite Fermi-temperature we obtain locality estimates that are uniform in the zero-temperature limit. A particular feature of all our results is that they depend only weakly on the point spectrum. Numerical tests confirm our analytical results.

This work extends and strengthens (Chen, Ortner 2016) and (Chen, Lu, Ortner 2018) for finite temperature models.
\end{abstract}

\section*{Introduction}\label{sec:introduction}\let\thefootnote\relax\footnote{CO is supported by EPSRC Grant EP/R043612/1 and Leverhulme Research Project Grant RPG-2017-191.\\
JT is supported by EPSRC as part of the MASDOC DTC, Grant No. EP/HO23364/1.\\
HC is supported by Thousand Talents Program for Young Professionals, and the Fundamental Research Funds for the Central Universities of China under grant 2017EYT22.\\
JT and CO gratefully acknowledge the hospitality of Peking Normal University during the work on this project.}A wide range of electronic, optical and magnetic properties of solids are determined by electronic structure. Computational methods, such as density functional theory, have been used successfully to model electronic structure and thus allowed the investigation and prediction of properties of materials \cite{bk:finnis,bk:martin04}. The tight binding model is a simple quantum mechanical model lying, both in terms of computational cost and accuracy, between empirical interatomic potential methods and expensive \textit{ab initio} calculations. Nevertheless, due to the underlying eigenvalue problem, a naive implementation of tight binding models requires $O(N^3)$ computational cost, where $N$ denotes the number of particles in the simulation. A possible route to alleviate this cost are linear scaling algorithms \cite{Kohn1996,Goedecker99,Niklasson2011,BowlerMiyazaki2012}, which rely on the ``nearsightedness'' of the density matrix. 

If the quantities of interest in a simulation are mechanical properties, then it may be advantageous to entirely bypass the electronic structure model and replace it with an interatomic potential (IP). The recent transfer of machine learning technology into this domain has made it possible to ``fit'' high-accuracy IP models \cite{Shapeev2016,Bartok2018,BartokCsanyiEtAl2010,BehlerParrinello2007}, which 
makes this approach particularly attractive. A starting assumption in most IP models for materials is that the potential energy surface can be decomposed into site energies, i.e., contributions from individual atoms that depend only on a small neighbourhood.

A partial justification for this assumption was given in \cite{ChenOrtner16,ChenLuOrtner18}, for a linear tight binding model at finite Fermi-temperature. For nuclei positions $y = \{y_n\}$, it was shown that there is a decomposition of the total potential energy into site energies
\begin{gather}\label{site_decomp} 
G(y) = \sum_{\ell} G_\ell(y), \qquad \text{where} \qquad 
\left|\frac{\partial G_\ell(y)}{\partial y_n} \right| \lesssim e^{-\eta |y_\ell -y_n|},
\end{gather}
for some $\eta > 0$. Similar estimates are also shown for higher derivatives. 

The exponent $\eta$ in \cref{site_decomp} measures the interatomic interaction range. Classical IPs are typically fairly short-ranged (using a cut-off on the order of 2-3 interatomic distances), which is only justified if $\eta$ is not too small. Similarly, in QM/MM multi-scale schemes the exponent $\eta$ determines the size of the QM region that must be imposed \cite{CsanyiAlbaretMorasPayneDeVita05,ChenOrtner16,ChenOrtner17} as well as the interaction range of the coarse-grained MM model. We emphasize that results such as \cref{site_decomp} do not follow from the classical near-sightedness of the density matrix. Indeed, exponential off-diagonal decay of the density matrix is not sufficient to validate multi-scale and hybrid models \cite{CsanyiAlbaretMorasPayneDeVita05,ChenOrtner17}.

Unfortunately, one expects (and we make this precise in the present paper) that, in general, 
\[
\eta \sim \beta^{-1},
\]
where $\beta$ is the inverse Fermi-temperature. This means that, for moderate to low temperature regimes, the practical value of \cref{site_decomp} is limited. 

The main purpose of the present paper is to demonstrate that for insulators the presence of a spectral gap  significantly improves the estimate. Specifically, we consider a linear tight-binding model at either zero or finite-Fermi temperature, with electrons in a grand-canonical ensemble. In this setting we prove that  \cref{site_decomp} holds with $\eta$ independent of $\beta$, but instead $\eta$ is linear in the spectral gap. Moreover, we demonstrate that  ``pollution'' of the spectral gap by a point spectrum caused, for example, by local defects in the crystal, affects only the prefactors, but not the exponent $\eta$ in \cref{site_decomp}. These results significantly strengthen the locality results of \cite{ChenOrtner16,ChenLuOrtner18} as well as extend them to the case of zero Fermi-temperature.

In addition to supporting the justification of interatomic potentials and QM/MM multi-scale models, our results also allow for an extension of the thermodynamic limit models for crystalline defects \cite{EhrlacherOrtnerShapeev16,ChenOrtner16,ChenNazarOrtner19} to the zero-temperature case and an investigation of the (non-trivial) relationship between zero and finite-temperature models, which we will pursue in a forthcoming paper \cite{inprep}.

\subsection*{Outline}
In \cref{sec:results}, we state the main results of this paper. In order to do this, we introduce a simple two-centre linear tight binding model (\cref{subsec:tight_binding_model}) and show that, at both finite and zero Fermi-temperature, the total energy of the system can be decomposed into exponentially localised site energy contributions (\cref{subsec:site_energy_decomposition}). We then discuss how these results can be improved upon in the case of a point defect embedded into a reference configuration (\cref{subsec:site_energy_decomposition}), showing that the resulting point spectrum does not affect the exponent. In \cref{sec:numerics} we provide numerical tests confirming our analytical results. The main conclusions of this work are then discussed in \cref{sec:conclusions} and all of the proofs are collected into \cref{sec:proofs}.

\subsection*{Notation}
The Frobenius norm will be denoted by $\|\cdot\|_\mathrm{F}$ while $\|\cdot\|$ and $|\cdot|$ will denote the $\ell^2$ and Euclidean norms, respectively. We write $b + A = \{ b + a \colon a \in A \}$ and similarly for $A - b$. Moreover, for $A \subset \mathbb C$ and $b \in \mathbb C$, the distance between $b$ and $A$ is defined by $\mathrm{dist}(b,A) \coloneqq \inf_A |b - \cdot\,|$. For a subset $A \subset \mathbb R$ and $\delta>0$, we denote the ball of radius $\delta$ about $A$ by $B_\delta(A) \coloneqq \{ r \in \mathbb R \colon \mathrm{dist}(r,A) < \delta \}$. For a finite set $A$, we denote by $\#A$ the cardinality of $A$. The set of strictly positive real numbers will be denoted by $\mathbb R_+ \coloneqq \{r \in \mathbb R \colon r > 0\}$. 

The symbol $C$ will denote a generic positive constant that may change from one line to the next. In calculations, $C$ will always be independent of Fermi-temperature. The dependencies of $C$ will normally be clear from context or stated explicitly.  

\section{Results}\label{sec:results}
\subsection{Tight Binding Model}
\label{subsec:tight_binding_model}

We consider a finite or countable reference configuration $\Lambda \subset \mathbb R^d$ and deformation $y \colon \Lambda\to\mathbb R^d$ satisfying the following uniform non-interpenetration condition:\newcommand{\asNonInter}{\textup{\textbf{(L)}}}
\begin{assumption}{\asNonInter}
There exists $\ctnoninterpen>0$ such that
$|y(\ell) - y(k)|\geq \ctnoninterpen$ for all $\ell, k \in \Lambda$.
\end{assumption}
We consider $\numorbitals$ atomic orbitals per atom, indexed by $1\leq a,b \leq \numorbitals$. For a given admissible configuration $y$, we define the following {two-centre tight binding} Hamiltonian:
\newcommand{\asTB}{\textup{\textbf{(TB)}}}
\begin{assumption}{\asTB}
For $\ell, k \in \Lambda$ and $1\leq a,b \leq\numorbitals$, we suppose that the Hamiltonian take the form
\begin{equation}\label{a:two-centre}
\Ham(y)_{\ell k}^{ab} = h^{ab}_{\ell k}\left(y(\ell) - y(k)\right)
\end{equation}
where $h^{ab}_{\ell k}\colon \mathbb R^d \to \mathbb R$ are $\ctHamregularity$ times continuously differentiable for some $\ctHamregularity\geq1$. Further, we assume that there exist 
$\ctTBprefactor{} \coloneqq (\ctTBprefactor{0},\dots,\ctTBprefactor{\nu})$, $\ctTBexponent{} \coloneqq (\ctTBexponent{0},\dots,\ctTBexponent{\nu})\in \left(\mathbb R_+\right)^{\nu + 1}$
such that, for each $1\leq j\leq \ctHamregularity$, 
\begin{equation}
\label{a:TB}	
\left|h^{ab}_{\ell k}(\xi)\right| \leq \ctTBprefactor{0}\,e^{-\ctTBexponent{0} |\xi|}
\quad \text{and} \quad 
\left|\partial^\alpha h_{\ell k}^{ab}(\xi)\right| \leq \ctTBprefactor{j} \, e^{-\ctTBexponent{j} |\xi|} 
\quad \forall \xi \in \mathbb R^d
\end{equation}
for all multi-indices $\alpha \in \mathbb N^d$ with $|\alpha| = j$. Finally, we suppose that $h^{ab}_{\ell k}(\xi) = h_{k\ell}^{ba}(-\xi)$ for all $\xi \in \mathbb R^d$ and $1\leq a,b\leq\numorbitals, \ell, k \in \Lambda$.
\end{assumption}

It is important to emphasise that the constants $\ctTBprefactor{},\ctTBexponent{} \in (\mathbb R_+)^{\ctHamregularity + 1}$ in \cref{a:TB} are chosen to be independent of the atomic sites $\ell, k \in \Lambda$.

The condition in \cref{a:TB} with $j = 0$ is satisfied for all linear tight binding models. In fact, in most tight binding models, a finite cut-off radius is used and so Hamiltonian entries are zero for atoms beyond a finite interaction range. For $j = 1$, \cref{a:TB} states that there are no long-range interactions. That is, the dependence of the Hamiltonian entry $\Ham(y)^{ab}_{\ell k}$ on site $m$ decays exponentially to zero in ${|y(\ell) - y(m)| + |y(m) - y(k)|}$. In particular, we are assuming that the Coulomb interactions have been screened.

Since the $h^{ab}_{\ell k}$ depend on the atomic sites, we allow for multi-lattice reference configurations with possibly multiple atomic species. While the assumptions are motivated by the lattice setting, we do not define exactly what we mean by $\Lambda$ and thus the presentation is kept abstract and the mathematical results are more general.

Under \asTB, $\sigma(\Ham(y))\subset[\underline{\sigma},\overline{\sigma}]$ where $\underline{\sigma},\overline{\sigma}$ only depend on $\ctnoninterpen,d,\ctTBprefactor{0},\ctTBexponent{0}$ and are independent of system size and configuration $y$ satisfying \asNonInter~with the constant $\ctnoninterpen$. A proof of this fact is an application of the Gershgorin circle theorem \cite[Lemma~4]{ChenOrtner16}.

\begin{remark}[Symmetries]
By \cref{a:two-centre}, if $\sigma$ is a permutation of $\Lambda$ leaving the individual atomic species invariant then we have $\Ham(y\circ\sigma)^{ab}_{\sigma^{-1}(\ell)\sigma^{-1}(k)} = \Ham(y)^{ab}_{\ell k}$.

In practice, we also require the Hamiltonian to be invariant under isometries of $\mathbb R^d$ up to an orthogonal change of basis \cite{SlaterKoster1954}: if $\mathcal I \colon \mathbb R^d \to \mathbb R^d$ is an isometry of $\mathbb R^d$, then there exists a block diagonal orthogonal matrix $Q$ such that $\Ham(\mathcal I \circ y) = Q\cdot \Ham(y)\cdot Q^T$. For a single atomic orbital per atom, this takes the form $\Ham(\mathcal I \circ y) = \Ham(y)$. This condition is derived in \cite[Appendix~A]{ChenOrtner16}.

Henceforth we will entirely ignore these symmetries, however we mention now that they give rise to permutation and isometry invariance of the site energies that we define below. This can be see exactly as in \cite{ChenOrtner16}.
\end{remark}

While nuclei are treated as classical particles, we assume that electrons are described by a grand canonical potential model. That is, the Fermi-temperature, volume and chemical potential, $\mu$, are fixed model parameters.  In this model, after diagonalising the Hamiltonian,
\begin{equation}\label{eq:diag_Ham}\Ham(y)\psi_s = \lambda_s\psi_s\quad\text{where}\quad\|\psi_s\| = 1,\end{equation}
(where the dependence of the eigenpair $(\lambda_s,\psi_s)$ on $y$ has been omitted) the potential energy surface is given by
\begin{equation}\label{grand-potential}
	G^\beta(y) \coloneqq \sum_{s} \mathfrak{g}^\beta(\lambda_s;\mu).
\end{equation}
Here, $\beta$ is the inverse Fermi-temperature given by $T = (k_\textrm{B}\beta)^{-1}$ where $k_\textrm{B}$ is the Boltzmann constant and $T$ the Fermi-temperature. We consider $\beta < \infty$ (as defined in \cite{ChenLuOrtner18,DavidMermin1965}) and $\beta = \infty$:
\begin{equation*}
	\mathfrak{g}^\beta(z;\mu) \coloneqq \frac{2}{\beta}\log\left( 1- f_\beta(z-\mu) \right) \quad \text{and} \quad 	\mathfrak{g}^\infty(z;\mu) \coloneqq 2(z-\mu)\chi_{(-\infty,\mu)}(z)
\end{equation*}
where $f_\beta \coloneqq (1 + \exp(\beta \,\cdot))^{-1}$ is the Fermi-Dirac distribution which describes the occupation numbers for the electronic states. The factor of $2$ accounts for the spin. The zero Fermi-temperature energy is simply given by the point-wise limit as $\beta \to \infty$. 

\begin{remark}
    Our analysis requires $G^\infty(y)$ to be a differentiable function of the configuration (that is, the derivative with respect to $[y(m)]_i$ exists for all $m \in \Lambda$ and $1\leq i \leq d$) and so we will usually impose the condition that $\mu \not \in \sigma(\Ham(y))$. Justification for considering this zero Fermi-temperature grand potential is given in \cite{ChenLuOrtner18} and a forthcoming paper \cite{inprep} where we formulate the geometry relaxation problem for the grand potential \cref{grand-potential} at zero Fermi-temperature as a variational problem and show that this is consistent with taking Fermi-temperature to zero.
\end{remark}

\subsection{Site Energy Decomposition}
\label{subsec:site_energy_decomposition}
For a given configuration, $y\colon \Lambda \to\mathbb R^d$ with $\Lambda$ finite, we can distribute the total energy of the system into {site energy} contributions. 
Since $\|\psi_s\| = 1$, we can decompose $G^\beta(y)$ into
\begin{equation}\label{eq:site_energy_decompostion}
G^\beta(y) = \sum_{\ell\in\Lambda} G_{\ell}^\beta(y) \quad \text{where} \quad 
G_{\ell}^{\beta}(y) \coloneqq \sum_s \mathfrak{g}^\beta(\lambda_s;\mu)\sum_a[\psi_s]_{\ell a}^2.
\end{equation} 
Using resolvent calculus, the site energies defined in \cref{eq:site_energy_decompostion} can be extended to the case where $\Lambda$ is infinite: By Lemma~\ref{lem:analytic-cont} below, $\mathfrak{g}^\beta(\,\cdot\,;\mu)$ extends to a holomorphic function defined on the set $\mathbb C \setminus \{\mu + ir \colon r \in \mathbb R, |r| \geq \pi\beta^{-1}\}$. Therefore, we may write
\begin{equation}\label{eq:site_energy_contour}
	G_{\ell}^\beta(y) \coloneqq -\frac{1}{2\pi i} \sum_a\oint_{\mathscr C_\beta} \mathfrak{g}^\beta(z;\mu) \left(\Ham(y) - z\right)^{-1}_{\ell\ell,aa} {\rm d}z
\end{equation}
where $\mathscr C_\beta$ is a simple closed contour contained within the region of holomorphicity of $\mathfrak{g}^\beta(z;\mu)$ and encircling the spectrum $\sigma(\Ham(y))$; see \Cref{fig}. We may choose $\mathscr C_\beta$ such that 
\begin{equation}\label{eq:dist_contour_existing}
\mathrm{dist}\left(z, \sigma(\Ham(y)) \cup \{\mu + ir \colon r\in \mathbb R, |r| \geq \pi\beta^{-1}\}\right) \geq \frac{\pi}{2\beta} \quad \text{for all } z \in \mathscr C_\beta.
\end{equation}
Since the {resolvent operator}, $\left(\Ham(y) - z\right)^{-1}$, is a well defined bounded linear operator for all $z \in \mathbb C \setminus \sigma(\Ham(y))$, the definition in \cref{eq:site_energy_contour} is valid for countable $\Lambda$. Obtaining a site energy on the infinite domain can also be derived by taking an appropriate sequence of finite domains $\Lambda_R$ and considering the thermodynamic limit of the site energies along this sequence \cite[Lemma~3.1]{ChenLuOrtner18}. This resolvent calculus approach for the tight binding model has been widely used \cite{ChenOrtner16,ChenLuOrtner18,ELu10,Goedecker1995}. 

For finite Fermi-temperature, the site energies defined in \cref{eq:site_energy_contour} are exponentially localised \cite{ChenLuOrtner18,ChenOrtner16} in the sense of Proposition~\ref{prop:existing_locality} below. The only difference between Proposition~\ref{prop:existing_locality} and \cite[Lemma~7]{ChenOrtner16} or \cite[Lemma~2.1]{ChenLuOrtner18} is that we explicitly track the $\beta$-dependent constants in the estimates.
\newpage\begin{prop}[Finite Fermi-Temperature Locality for Metals]\label{prop:existing_locality}
\text{ }
\begin{enumerate}[label=(\roman*)]
\item Suppose $y\colon \Lambda\to \mathbb R^d$ and $\Ham(y)$ satisfy \asNonInter~and \asTB,~respectively. Then, for $1 \leq j \leq \ctHamregularity$, there exist positive constants $C_j = C_j(\beta)$ and $\eta_j = \eta_j(\beta)$ such that
	\begin{gather}\label{eq:locality_beta}
	\left| \frac{\partial^j G_{\ell}^\beta(y)}{\partial [y(m_1)]_{i_1} \dots \partial [y(m_j)]_{i_j}} \right| \leq C_j e^{-\eta_j\sum_{l=1}^j |y(\ell) - y(m_l)|} 
	\end{gather}
	for any $\ell,m_1,\dots,m_j \in \Lambda$ and $1\leq i_1,\dots,i_j \leq d$. 
	
\item For all sufficiently large $\beta$, $C_{j}(\beta) = C\beta^{\alpha}$ where $C>0$ depends only on $\numorbitals,\ctTBprefactor{},\ctnoninterpen$ and $\alpha>0$ depends only on $j$ and $d$. Further, $\eta_j(\beta) = c\min\{1,\beta^{-1}\}$ for some $c>0$ depending only on $j, \ctTBprefactor{},\ctTBexponent{},\ctnoninterpen$ and $d$.
\end{enumerate}
\end{prop}
\begin{proof}[Sketch of the Proof]
The $\beta$-dependence in the estimate \cref{eq:locality_beta} comes from the fact that the distance between the contour and the spectrum can, in general, only be bounded below by a constant multiple of $\beta^{-1}$ as in \cref{eq:dist_contour_existing}. We summarise the main ideas in \cref{proofs:existing_locality} below.
\end{proof}
For the case of insulators (where $\mu$ lies in a spectral gap), Proposition~\ref{prop:existing_locality} can both be improved and extended to zero Fermi-temperature. In this case, the following constants are strictly positive
\begin{align}\label{eq:insulator_distance}
	\ctgapfinitetemp(y) &\coloneqq \mathrm{dist}\left(\mu,\sigma(\Ham(y))\right)\quad \text{and}\\
	\ctgap(y) &\coloneqq \inf\left(\sigma(\Ham(y))\cap(\mu,+\infty)\right) - \sup\left(\sigma(\Ham(y))\cap(-\infty,\mu)\right). \label{eq:insulator_distance_2}
\end{align}
Using \cref{eq:insulator_distance_2}, and the fact that $z \mapsto 2(z-\mu)$ is analytic, means that the expression \cref{eq:site_energy_contour} holds for zero Fermi-temperature with a simple closed contour $\mathscr C_\infty$ encircling $\sigma(\Ham(y))\cap (-\infty,\mu)$ and avoiding $\sigma(\Ham(y))\cap (\mu,\infty)$; see \Cref{fig}. Further, we may suppose that
\begin{equation}\label{eq:distance_contours_1}
\textrm{dist}\left(z, \sigma(\Ham(y))\right) \geq \tfrac{1}{2}\ctgap(y) \quad \forall z \in \mathscr C_\infty.
\end{equation}
\begin{figure}[htb]
	\vspace{-8em}
	\centering
	\resizebox{\columnwidth} {!} {
		\tikzset{every picture/.style={line width=0.75pt}} 

\begin{tikzpicture}[x=0.75pt,y=0.75pt,yscale=-1,xscale=1]

\draw  [color={rgb, 255:red, 123; green, 123; blue, 123 }  ,draw opacity=1 ][dash pattern={on 5.63pt off 4.5pt}][line width=1.5]  (303.53,44.6) .. controls (303.53,27.14) and (317.69,12.98) .. (335.15,12.98) .. controls (352.61,12.98) and (366.77,27.14) .. (366.77,44.6) .. controls (366.77,62.06) and (352.61,76.22) .. (335.15,76.22) .. controls (317.69,76.22) and (303.53,62.06) .. (303.53,44.6) -- cycle ;
\draw [color={rgb, 255:red, 0; green, 0; blue, 0 }  ,draw opacity=1 ][line width=1.5]  [dash pattern={on 1.69pt off 2.76pt}]  (334.5,12) -- (334.5,279) ;

\draw  [color={rgb, 255:red, 34; green, 4; blue, 230 }  ,draw opacity=1 ][line width=1.5]  (38,126.69) .. controls (38,97.22) and (61.89,73.33) .. (91.36,73.33) -- (269.04,73.33) .. controls (298.51,73.33) and (322.4,97.22) .. (322.4,126.69) -- (322.4,159.97) .. controls (322.4,189.44) and (298.51,213.33) .. (269.04,213.33) -- (91.36,213.33) .. controls (61.89,213.33) and (38,189.44) .. (38,159.97) -- cycle ;
\draw [line width=1.5]    (86.9,135.96) -- (86.9,150.96) ;

\draw [color={rgb, 255:red, 123; green, 123; blue, 123 }  ,draw opacity=1 ]   (363.4,149.53) -- (277.9,149.95) ;
\draw [shift={(275.9,149.96)}, rotate = 359.72] [fill={rgb, 255:red, 123; green, 123; blue, 123 }  ,fill opacity=1 ][line width=0.75]  [draw opacity=0] (8.93,-4.29) -- (0,0) -- (8.93,4.29) -- cycle    ;
\draw [shift={(365.4,149.52)}, rotate = 179.72] [fill={rgb, 255:red, 123; green, 123; blue, 123 }  ,fill opacity=1 ][line width=0.75]  [draw opacity=0] (8.93,-4.29) -- (0,0) -- (8.93,4.29) -- cycle    ;
\draw  [color={rgb, 255:red, 255; green, 0; blue, 0 }  ,draw opacity=1 ][line width=1.5]  (640.84,143.96) .. controls (641.84,416.96) and (389.84,207.96) .. (336.34,207.24) .. controls (282.84,206.52) and (14.84,393.96) .. (14.84,141.96) .. controls (14.84,-110.04) and (281.84,81.96) .. (334.84,81.96) .. controls (387.84,81.96) and (639.84,-129.04) .. (640.84,143.96) -- cycle ;
\draw [line width=2.25]    (324.3,232.2) -- (344,255) ;

\draw [line width=2.25]    (323.3,253.2) -- (345.65,235.1) ;

\draw [line width=2.25]    (324.3,33.2) -- (344,56) ;

\draw [line width=2.25]    (322.97,53.65) -- (345.32,35.55) ;

\draw [color={rgb, 255:red, 123; green, 123; blue, 123 }  ,draw opacity=1 ][line width=1.5]  [dash pattern={on 1.69pt off 2.76pt}]  (251.34,115.52) -- (251.3,175.5) ;

\draw [color={rgb, 255:red, 123; green, 123; blue, 123 }  ,draw opacity=1 ][line width=1.5]  [dash pattern={on 1.69pt off 2.76pt}]  (431,142) -- (430,176) ;

\draw [color={rgb, 255:red, 123; green, 123; blue, 123 }  ,draw opacity=1 ]   (427.9,165.97) -- (252.15,166.59) ;
\draw [shift={(250.15,166.6)}, rotate = 359.8] [fill={rgb, 255:red, 123; green, 123; blue, 123 }  ,fill opacity=1 ][line width=0.75]  [draw opacity=0] (8.93,-4.29) -- (0,0) -- (8.93,4.29) -- cycle    ;
\draw [shift={(429.9,165.96)}, rotate = 179.8] [fill={rgb, 255:red, 123; green, 123; blue, 123 }  ,fill opacity=1 ][line width=0.75]  [draw opacity=0] (8.93,-4.29) -- (0,0) -- (8.93,4.29) -- cycle    ;
\draw  [fill={rgb, 255:red, 0; green, 0; blue, 0 }  ,fill opacity=1 ] (116,142) -- (251.3,142) -- (251.3,145.2) -- (116,145.2) -- cycle ;
\draw  [fill={rgb, 255:red, 0; green, 0; blue, 0 }  ,fill opacity=1 ] (431,142) -- (566.3,142) -- (566.3,145.2) -- (431,145.2) -- cycle ;
\draw [line width=1.5]    (274.9,135.96) -- (274.9,150.96) ;

\draw [line width=1.5]    (584.9,135.96) -- (584.9,150.96) ;

\draw  [color={rgb, 255:red, 123; green, 123; blue, 123 }  ,draw opacity=1 ][dash pattern={on 5.63pt off 4.5pt}][line width=1.5]  (303.86,244.15) .. controls (303.86,226.69) and (318.01,212.53) .. (335.47,212.53) .. controls (352.94,212.53) and (367.09,226.69) .. (367.09,244.15) .. controls (367.09,261.61) and (352.94,275.77) .. (335.47,275.77) .. controls (318.01,275.77) and (303.86,261.61) .. (303.86,244.15) -- cycle ;
\draw [color={rgb, 255:red, 123; green, 123; blue, 123 }  ,draw opacity=1 ]   (331.9,121.98) -- (253.15,122.58) ;
\draw [shift={(251.15,122.6)}, rotate = 359.56] [fill={rgb, 255:red, 123; green, 123; blue, 123 }  ,fill opacity=1 ][line width=0.75]  [draw opacity=0] (8.93,-4.29) -- (0,0) -- (8.93,4.29) -- cycle    ;
\draw [shift={(333.9,121.96)}, rotate = 179.56] [fill={rgb, 255:red, 123; green, 123; blue, 123 }  ,fill opacity=1 ][line width=0.75]  [draw opacity=0] (8.93,-4.29) -- (0,0) -- (8.93,4.29) -- cycle    ;
\draw [color={rgb, 255:red, 123; green, 123; blue, 123 }  ,draw opacity=1 ]   (337.6,137.26) -- (363.25,137.43) ;
\draw [shift={(365.25,137.44)}, rotate = 180.37] [fill={rgb, 255:red, 123; green, 123; blue, 123 }  ,fill opacity=1 ][line width=0.75]  [draw opacity=0] (8.93,-4.29) -- (0,0) -- (8.93,4.29) -- cycle    ;
\draw [shift={(335.6,137.25)}, rotate = 0.37] [fill={rgb, 255:red, 123; green, 123; blue, 123 }  ,fill opacity=1 ][line width=0.75]  [draw opacity=0] (8.93,-4.29) -- (0,0) -- (8.93,4.29) -- cycle    ;
\draw [line width=1.5]    (365.9,135.96) -- (365.9,150.96) ;

\draw [line width=1.5]    (396.9,135.96) -- (396.9,150.96) ;

\draw (44,230) node [scale=1.2,color={rgb, 255:red, 255; green, 0; blue, 0 }  ,opacity=1 ]  {$\mathscr{C}_{\beta }$};
\draw (67,192) node [scale=1.2,color={rgb, 255:red, 0; green, 0; blue, 255 }  ,opacity=1 ]  {$\mathscr{C}_{\infty }$};
\draw (304,143) node [scale=0.8,color={rgb, 255:red, 4; green, 0; blue, 0 }  ,opacity=1 ]  {$\mathsf{g}( y)$};
\draw (368,159) node [scale=0.8,color={rgb, 255:red, 0; green, 0; blue, 0 }  ,opacity=1 ]  {$\mathsf{g}^{\mathrm{ref}}$};
\draw (287,116) node [scale=0.8,color={rgb, 255:red, 0; green, 0; blue, 0 }  ,opacity=1 ]  {$\mathsf{d}^{\mathrm{ref}}$};
\draw (351,129) node [scale=0.8,color={rgb, 255:red, 0; green, 0; blue, 0 }  ,opacity=1 ]  {$\mathsf{d}( y)$};

\end{tikzpicture} }%
	\vspace{-8em}
	\caption{Cartoon depicting an approximation of $\sigma(\Ham(y))$ for $y \in \mathrm{Adm}(\Lambda)$ (in black on the real axis), see Lemma~\ref{lem:decomp-spec}, and the contours $\mathscr C_\beta$ (in red) and $\mathscr C_\infty$ (in blue). The positive constants $\ctgapfinitetemp(y),\ctgapfinitetempHom,\ctgap(y)$ and $\ctgapHom$ are also displayed. By \cref{eq:distance_singularity_beta}, the finite Fermi-temperature contours avoid the balls of radius $\ctAwayFromSingularity\beta^{-1}$ (grey dashed circles) about $\mu \pm i\pi\beta^{-1}$ (shown with black crosses). The spectrum pictured in Figure~\ref{fig} is qualitatively similar to that resulting from point defects in lattice structures.\label{fig}}
\end{figure}
In this case, we obtain the following locality results that are uniform in Fermi-temperature:
\begin{prop}[Locality Estimates for Insulators]\label{prop:locality_insulators}\text{ }
\begin{enumerate}[label=(\roman*)]
	\item Suppose $y \colon \Lambda \to \mathbb R^d$ and $\Ham(y)$ satisfy \asNonInter~and \asTB, respectively. Further, we assume that $\mu \not \in \sigma(\Ham(y))$. Then, for $1\leq j \leq \ctHamregularity$, there exist positive constants $C_{j}, \eta_{j}$, such that 
	\begin{gather}\label{eq:locality}
	\left| \frac{\partial^j G^\beta_{\ell}(y)}{\partial [y(m_1)]_{i_1} \dots \partial [y(m_j)]_{i_j}} \right| \leq C_{j}\,e^{-\eta_{j}\sum_{l=1}^j|y(\ell) - y(m_l)|} 
	\end{gather}
	for any $\ell,m_1,\dots,m_j \in \Lambda$ and $1\leq i_1,\dots,i_j \leq d$. 	

	\item For $\beta<\infty$, the constants $C_{j}$ and $\eta_{j}$ depend on $\ctgapfinitetemp = \ctgapfinitetemp(y)$ from \cref{eq:insulator_distance}. For all sufficiently small $\ctgapfinitetemp$, $C_{j}(\ctgapfinitetemp) = C\ctgapfinitetemp^\alpha$ where $C>0$ depends only on $\numorbitals,\ctTBprefactor{},\ctnoninterpen$ and $\alpha>0$ depends only on $j$ and $d$. Further, $\eta_{j}(\ctgapfinitetemp) = c \min\{1,\ctgapfinitetemp\}$ for some $c>0$ depending only on $j,\ctTBprefactor{},\ctTBexponent{},\ctnoninterpen$ and $d$.

	\item {For $\beta= \infty$, the constants $C_{j}$ and $\eta_{j}$ depend on $\ctgap = \ctgap(y)$ from \cref{eq:insulator_distance_2}. For all sufficiently small $\ctgap$, $C_{j}(\ctgap) = C\ctgap^\alpha$ where $C>0$ depends only on $\numorbitals,\ctTBprefactor{},\ctnoninterpen$ and $\alpha>0$ depends only on $j$ and $d$. Further, $\eta_{j}(\ctgap) = c \min\{1,\ctgap\}$ for some $c>0$ depending only on $j,\ctTBprefactor{},\ctTBexponent{},\ctnoninterpen$ and $d$.}
\end{enumerate}
\end{prop}
\begin{proof}[Sketch of the Proof]
This result follows from the same arguments as in the proof of Proposition~\ref{prop:existing_locality}. In this case, the pre-factors and exponents are $\beta$-independent because the constants \cref{eq:insulator_distance} and \cref{eq:insulator_distance_2} are. Again, the main ideas are summarised in \cref{proofs:existing_locality}.
\end{proof}

\subsection{Point Defects}
\label{subsec:point_defects}
We suppose that a given reference configuration has a spectral gap and that the chemical potential is fixed within the gap. Then we show that point defect configurations introduce additional ``defect states" into the system and also perturbs the essential spectrum. The main result of this paper is that, within this setting, the locality results discussed in \cref{subsec:site_energy_decomposition} are independent of discrete spectra inside the band gap in the sense of Theorems~\ref{thm:improved_locality_beta} and \ref{thm:improved_locality} below.

\subsubsection{Reference Configurations}
\label{subsubsec:band_structure_homogeneous}
In preparation for this result, we consider the Hamiltonian on a fixed reference configuration, $\Lambda^\mathrm{ref}$, given by $\left(\Ham^\mathrm{ref}\right)_{\ell k}^{ab} = h_{\ell k}^{ab}(\ell - k)$ for $\ell,k \in \Lambda^\mathrm{ref}$ and $1\leq a,b\leq\numorbitals$. In order to keep the presentation abstract and the mathematical results general we will not explicitly define $\Lambda^\mathrm{ref}$, but we will always be thinking of a multi-lattice.

For the remainder of this paper, we make the following assumption:
\newcommand{\asInsulator}{\textbf{(GAP)}}
\begin{assumption}{\asInsulator}
There exists a band gap in $\sigma(\Ham^{\mathrm{ref}})$ and the chemical potential, $\mu$, lies in the interior of this band gap. 
\end{assumption}
Under \asInsulator, we may introduce the following positive constants that will determine the interaction range in the improved locality estimates:
\begin{align}
\label{eq:gapfinite}
\ctgapfinitetempHom &\coloneqq \mathrm{dist}\left(\mu,\sigma(\Ham^\mathrm{ref})\right) \quad \text{and}\\
\label{eq:gapzero}\ctgapHom &\coloneqq \inf\left( \sigma(\Ham^\mathrm{ref})\cap (\mu,+\infty)\right) - \sup\left( \sigma(\Ham^\mathrm{ref})\cap (-\infty,\mu)\right).
\end{align}

\begin{remark}
In the case where $\Lambda^\mathrm{ref}$ is a multi-lattice formed by taking the union of finitely many copies of a Bravais lattice, $\mathbb{L}$, we may apply Bloch's Theorem \cite{bk:kittel} to diagonalise $\Ham^\mathrm{ref}$ and write the spectrum as the union of finitely many continuous energy bands:
\begin{equation}
\label{multilattice-bands}
\sigma(\Ham^\mathrm{ref}) = \bigcup_{\alpha} \{ \varepsilon^\alpha(k) \colon k \in \textrm{BZ} \}
\end{equation}
where $\varepsilon^{\alpha} \colon \overline{\textrm{BZ}} \to \mathbb R$ are continuous functions on the (closure of the) Brillouin zone.
\end{remark}

\subsubsection{Point Defect Reference Configurations}
\label{subsubsec:point_defect_reference_configurations}
 
\newcommand{\asPointDefect}{\textbf{(P)}}
From now on, we shall assume that $\Lambda$ is a {point defect reference configuration}:
\begin{assumption}{\asPointDefect}
	Given a reference configuration $\Lambda^\mathrm{ref} \subset \mathbb R^d$, we suppose $\Lambda\subset \mathbb R^d$ is such that there exists a positive constant $R_\mathrm{def}$ with $\Lambda^\mathrm{ref} \setminus B_{R_\textrm{def}} = \Lambda \setminus B_{R_\textrm{def}}$ and $\#\left(B_{R_\textrm{def}}\cap\Lambda\right)<\infty$.
\end{assumption}
We now introduce energy spaces of {displacements} which restricts the class of admissible configurations. Given $\ell \in \Lambda$ and $\rho \in \Lambda - \ell$, we define the finite difference $D_\rho u(\ell) \coloneqq u(\ell + \rho) - u(\ell)$. The full (infinite) finite difference stencil is then defined to be $Du(\ell) \coloneqq \left( D_\rho u(\ell)\right)_{\rho\in\Lambda - \ell}$. For $\ctGamma > 0$, the $\ell^2_\ctGamma$ semi-norm on the full interaction stencil is given by 
\[ 
\|Du\|_{\ell^2_\ctGamma} \coloneqq \left(\sum_{\ell \in \Lambda} \sum_{\rho\in\Lambda-\ell} e^{-2\ctGamma|\rho|}|D_\rho u(\ell)|^2\right)^{1/2}.
\]
All of the semi-norms $\|D\cdot\|_{\ell^2_\ctGamma}$ for $\ctGamma > 0$ are equivalent \cite{ChenNazarOrtner19} and so we fix an exponent $\ctGamma>0$ for the remainder of this paper and define the following function space of finite energy displacements:
\begin{equation*} 
\W(\Lambda) \coloneqq \{ u \colon \Lambda \to \mathbb R^d \colon \|Du\|_{\ell^2_\ctGamma} < \infty \}.  
\end{equation*}
We hence define the space of {admissible configurations} by
\begin{equation*}
\textrm{Adm}(\Lambda) \coloneqq \{ y \in x + \W(\Lambda)\colon y \text{ satisfies \asNonInter}\}
\end{equation*}
where $x\colon\Lambda\to\Lambda$ denotes the identity configuration.

For $y \in \mathrm{Adm}(\Lambda)$, the spectrum, $\sigma(\Ham(y))$, can be related to $\sigma(\Ham^\mathrm{ref})$: 
\begin{lemma}[Decomposition of the Spectrum]
\label{lem:decomp-spec}
Fix $y \in \mathrm{Adm}(\Lambda)$. Then, for all $\delta>0$, there exists $R_\delta>0$ such that $\#\left(\sigma(\Ham(y)) \setminus B_\delta(\sigma(\Ham^\mathrm{ref})) \right) \leq R_\delta$.
\end{lemma}
\begin{proof}[Sketch of the Proof]
We approximate $\Ham(y)$ as a finite rank update of $\Ham^\mathrm{ref}$. The finite rank perturbation only introduces finitely many eigenvalues outside the essential spectrum of $\Ham^\mathrm{ref}$. See \cref{proofs:decomposition_spectrum} for a full proof.
\end{proof}
\begin{remark}
    In the case that $\Lambda^\mathrm{ref}$ is a multi-lattice, $\sigma(\Ham^\mathrm{ref})$ is banded in the sense of \cref{multilattice-bands}. This means Lemma~\ref{lem:decomp-spec} states that point defects give rise to a finite number of ``defect states'' located away from the spectral bands as depicted in Figure~\ref{fig}.
\end{remark}

\subsubsection{Locality of Site Energies}
\label{subsubsec:locality_site_energies}
Lemma~\ref{lem:decomp-spec}, together with a locality result for the spectral projection (see Lemma~\ref{lem:Thomas-Combes} below) corresponding to the finitely many eigenvalues bounded away from the spectral bands, allows us to approximate $(\Ham(y) - z)^{-1}$ in terms of the reference resolvent, $(\Ham^\mathrm{ref} - z)^{-1}$. This means we can apply the existing locality estimates of Proposition~\ref{prop:locality_insulators} on the reference spectrum. The approximation does not affect the exponent in the estimates and only increases the constant pre-factor. We show that the pre-factor may be chosen to depend on the atomic sites and this converges exponentially to the corresponding pre-factor in the defect-free case, as we send the atomic sites away from the defect core. That is, away from the defect, the locality estimates resemble the corresponding estimates for the reference configuration.

\newpage\begin{theorem}[Improved Finite Fermi-Temperature Locality]\label{thm:improved_locality_beta}\text{ }
	\begin{enumerate}[label=(\roman*)]
		\item Fix $y \in \mathrm{Adm}(\Lambda)$. Then, for $1\leq j \leq \ctHamregularity$, $\ell \in \Lambda$, $\bm{m} = (m_1,\dots,m_j) \in \Lambda^j$ and $1\leq i_1,\dots,i_j\leq d$, there exists a positive constant $C_{\beta j} = C_{\beta j}(\ell, \bm{m})$ such that 
		\begin{gather}\label{eq:locality_finite}
		\left| \frac{\partial^j G_{\ell}^\beta(y)}{\partial [y(m_1)]_{i_1} \dots \partial [y(m_j)]_{i_j}} \right| \leq C_{\beta j} e^{-\eta_{j}^\mathrm{ref}\sum_{l=1}^j |y(\ell) - y(m_l)|} 
		\end{gather}
		where $\eta_{j}^\mathrm{ref} \coloneqq c \min\{1,\ctgapfinitetempHom\}$ and $c>0$ depends only on $j,\ctTBprefactor{},\ctTBexponent{},\ctnoninterpen$ and $d$.
	
		\item $C_{\beta j}(\ell,\bm{m})$ is uniformly bounded independently of $\ell$ and $\bm{m}$. Let $C_{\beta j}^\mathrm{ref} \coloneqq C_{\beta j}$ when $\Lambda = \Lambda^\mathrm{ref}$ and $y = x$. If $\ell,m_{1},\dots,m_j \in B_R(\xi)$ for some $R>0$, then $C_{\beta j}(\ell,\bm{m}) \to C_{\beta j}^\mathrm{ref}$ as $|\xi|\to \infty$, with an exponential rate.
	
		\item If $\mu \not\in \sigma(\Ham(y))$, then $C_{\beta j}$ can be chosen to be $\beta$-independent. On the other hand, if $\mu \in \sigma(\Ham(y))$, then $C_{\beta j} = C\beta^{j-1}$ for some $C>0$ depending only on $j, \ctTBprefactor{},\ctTBexponent{},\numorbitals,\ctnoninterpen$ and $d$. 
	\end{enumerate}	
\end{theorem}
\begin{remark}
By \textit{(iii)}, if $\mu \in \sigma(\Ham(y))$, then the first derivatives of the site energies are uniformly bounded in $\beta$ despite the fact that the zero Fermi-temperature site energies are not differentiable. In fact, the point-wise limit of the first derivatives as $\beta \to \infty$ exist. 
\end{remark}
We have the following analogous zero Fermi-temperature result:
\begin{theorem}[Improved Zero Fermi-Temperature Locality]\label{thm:improved_locality}\text{ }
\begin{enumerate}[label=(\roman*)]
		\item Fix $y \in \mathrm{Adm}(\Lambda)$ with $\mu \not\in \sigma(\Ham(y))$. Then, for $1\leq j \leq \ctHamregularity$, $\ell \in \Lambda, \bm{m} = (m_1,\dots,m_j) \in \Lambda^j$, and $1\leq i_1,\dots,i_j \leq d$, there exists a positive constant $C_{\infty j}(\ell,\bm{m})$ such that 
		\begin{gather}\label{eq:locality_zero}
		\left| \frac{\partial^j  G^\infty_{\ell}(y)}{\partial [y(m_1)]_{i_1} \dots \partial [y(m_j)]_{i_j}} \right| \leq C_{\infty j}(\ell,\bm{m}) \,e^{-\eta^\mathrm{ref}_{\infty j}\sum_{l=1}^j |y(\ell) - y(m_l)|} 
		\end{gather}
		where $\eta^\mathrm{ref}_{\infty j} \coloneqq c \min\{1,\ctgapHom\}$ and $c > 0$ depends only on $j, \ctTBprefactor{},\ctTBexponent{},\ctnoninterpen$ and $d$.
	
		\item $C_{\infty j}(\ell,\bm{m})$ is uniformly bounded independently of $\ell$ and $\bm{m}$. Let $C_{\infty j}^\mathrm{ref}\coloneqq C_{\infty j}$ when $\Lambda = \Lambda^\mathrm{ref}$ and $y = x$. If $\ell, m_1, \dots,m_j\in B_R(\xi)$ for some $R>0$, then $C_{\infty j}(\ell,\bm{m}) \to C_{\infty j}^\mathrm{ref}$ as $|\xi|\to\infty$, with an exponential rate. 
	\end{enumerate}
\end{theorem} 
\begin{remark}
	In the case $j=1$, Theorem~\ref{thm:improved_locality_beta} part \textit{(ii)} takes the form\[ |C_{\beta 1}(\ell,{m}) - C_{\beta 1}^\mathrm{ref}| \lesssim e^{-\eta_{1}^\mathrm{ref}(|y(\ell)| + |y(m)| - |y(\ell) - y(m)|)}.\]Similarly for Theorem~\ref{thm:improved_locality} with $\beta = \infty$ and the exponent $\eta_{\infty1}^\mathrm{ref}$. For higher derivatives, the relationship between $\ell$ and $\bm{m}$ is more complicated.
\end{remark}


\newcommand{\chc}[1]{{\color{blue} \it \small [HC: #1]}}
\newcommand{\hc}[1]{{\color{blue} #1}}

\section{Numerical Tests}
\label{sec:numerics}
\setcounter{equation}{0}

\def\fc{f_{\rm c}}
\def\Rc{R_{\rm c}}

In this section, we present numerical simulations to support our analytical results. We use a practical tight binding model, the NRL model \cite{cohen94,mehl96,papaconstantopoulos15}, to test the force-locality in bulk carbon and silicon, both with and without an interstitial defect. Since we are unaware of established codes that compute site energies and their derivatives, we implemented these models in the {\tt Julia} package {\tt SKTB.jl} \cite{gitSKTB}.

\subsection{The NRL tight binding model}
\label{sec:NRL}
The NRL tight binding model, developed by Cohen, Mehl, and Papaconstantopoulos \cite{cohen94},
is slightly more general than our formulation in 
\S~\ref{sec:results}. It is non-orthogonal, which means that the energy  levels are now determined by the generalised eigenvalue problem
\begin{eqnarray}\label{NRL:diag_Ham}
\Ham(y)\psi_s = \lambda_s\mathcal{M}(y)\psi_s 
\qquad\text{where}\quad \psi_s^{\rm T}\mathcal{M}(y)\psi_s = 1,
\end{eqnarray}
which has an additional overlap matrix $\mathcal{M}(y)$ compared with \eqref{eq:diag_Ham}.
Furthermore, the NRL hamiltonian and overlap matrices are construct both from hopping elements as in \eqref{a:two-centre} 
as well as on-site matrix elements as a function of the local environment. For carbon and silicon they are parameterised as follows (for other elements the parameterisation is similar): 

To define the on-site terms, each atom $\ell$ is assigned a pseudo-atomic density 
\begin{eqnarray*}
\rho_{\ell} := \sum_{k}e^{-\lambda^2 r_{\ell k}}  \fc(r_{\ell k}),
\end{eqnarray*}
where the sum is over all of the atoms $k$ within the cutoff $\Rc$ of atom $\ell$,
$\lambda$ is a fitting parameter, $\fc$ is a cutoff function
\begin{eqnarray*}
    \fc(r) = \frac{\theta(\Rc-r)}{1+\exp\big((r-\Rc)/l_c + L_c\big)} ,
\end{eqnarray*}
with $\theta$ the step function, and the parameters $l_c=0.5$, $L_c=5.0$ for most elements.
Although, in principle, the on-site terms should have off-diagonal elements, the NRL model follows traditional practice and only include the diagonal terms.
Then, the on-site terms for each atomic site $\ell$ are given by
\begin{eqnarray}\label{NRLonsite}
\Ham(y)_{\ell\ell}^{\upsilon\upsilon}
:= a_{\upsilon} + b_{\upsilon}\rho_{\ell}^{2/3} + c_{\upsilon}\rho_{\ell}^{4/3} + d_{\upsilon}\rho_{\ell}^{2},
\end{eqnarray}
where $\upsilon=s,p$, or $d$ is the index for angular-momentum-dependent atomic orbitals and $(a_\upsilon)$, $(b_\upsilon)$, $(c_\upsilon)$, $(d_\upsilon)$ are fitting parameters. 
The on-site elements for the overlap matrix are simply taken to be the identity matrix.

The off-diagonal NRL Hamiltonian entries follow the formalism of 
Slater and Koster who showed in \cite{SlaterKoster1954} that all two-centre (spd) hopping integrals can be constructed from ten independent ``bond integral'' parameters $h_{\upsilon\upsilon'\mu}$, where
\begin{eqnarray*}
(\upsilon\upsilon'\mu) = ss\sigma,~sp\sigma,~pp\sigma,~pp\pi,
~sd\sigma,~pd\sigma,~pd\pi,~dd\sigma,~dd\pi,~{\rm and}~dd\delta.
\end{eqnarray*}
The NRL bond integrals are given by
\begin{eqnarray}\label{NRLhopping}
h_{\upsilon\upsilon'\mu}(r) 
:= \big(e_{\upsilon\upsilon'\mu} + f_{\upsilon\upsilon'\mu}r + g_{\upsilon\upsilon'\mu} r^2 \big) e^{-h_{\upsilon\upsilon'\mu}r} \fc(r)
\end{eqnarray}
with fitting parameters $e_{\upsilon\upsilon'\mu}, f_{\upsilon\upsilon'\mu}, 
g_{\upsilon\upsilon'\mu}, h_{\upsilon\upsilon'\mu}$. The matrix elements
$\Ham(y)_{\ell k}^{\upsilon\upsilon'}$ are constructed from
the $h_{\upsilon\upsilon'\mu}(r)$ by a standard procedure~\cite{SlaterKoster1954}.

The analogous bond integral parameterisation of the overlap matrix 
is given by  
\begin{eqnarray}\label{NRLhopping-M}
m_{\upsilon\upsilon'\mu}(r) 
:= \big(\delta_{\upsilon\upsilon'} + p_{\upsilon\upsilon'\mu} r + q_{\upsilon\upsilon'\mu} r^2 + r_{\upsilon\upsilon'\mu} r^3 \big) e^{-s_{\upsilon\upsilon'\mu} r} \fc(r)
\end{eqnarray}
with the fitting parameters $(p_{\upsilon\upsilon'\mu}), (q_{\upsilon\upsilon'\mu}), (r_{\upsilon\upsilon'\mu}), (s_{\upsilon\upsilon'\mu})$ and $\delta_{\upsilon\upsilon'}$ the Kronecker delta function. 

The fitting parameters in the foregoing expressions 
are determined by fitting to some high-symmetry first-principle calculations: In the NRL method, a database of eigenvalues (band structures) and total energies were constructed for several crystal structures at  several volumes. Then the parameters are chosen such that the eigenvalues and energies in the database are reproduced.
For practical simulations, the parameters for different elements can be found in \cite{papaconstantopoulos15}.

\subsection{Test systems}
\label{sec:carbon}
Our two test systems are diamond cubic bulk carbon and bulk silicon, which provide ideal test cases of our theory due to their clearly defined band gaps. Since carbon has a much larger band gap than silicon we will also be able to test how this affects locality of interaction.

For both elements, we simulate a supercell model consisting of $5 \times 5 \times 5$ diamond cubic unit cells, containing 1000 atoms in total. First, we use the NRL tight binding model to relax the cells to their ground states (this only  rescales the cells but does not change  their shape). 
We then compute the band structures which are, respectively, shown in Figures~\ref{fig:C_band} and~\ref{fig:Si_band}. 
We verified our implementation by comparing the band structure for the silicon model against that published in \cite{papaconstantopoulos97}.
The Fermi energy is  chosen to be the mid point between the highest occupied state and the lowest unoccupied state of the homogeneous 1000-atom system. For both systems we observe clearly defined band gaps around the Fermi energy, approximately 0.98 eV for Si and 3.83 eV for C.

\begin{figure}[htb]
	\centering 
	\includegraphics[width=15cm]{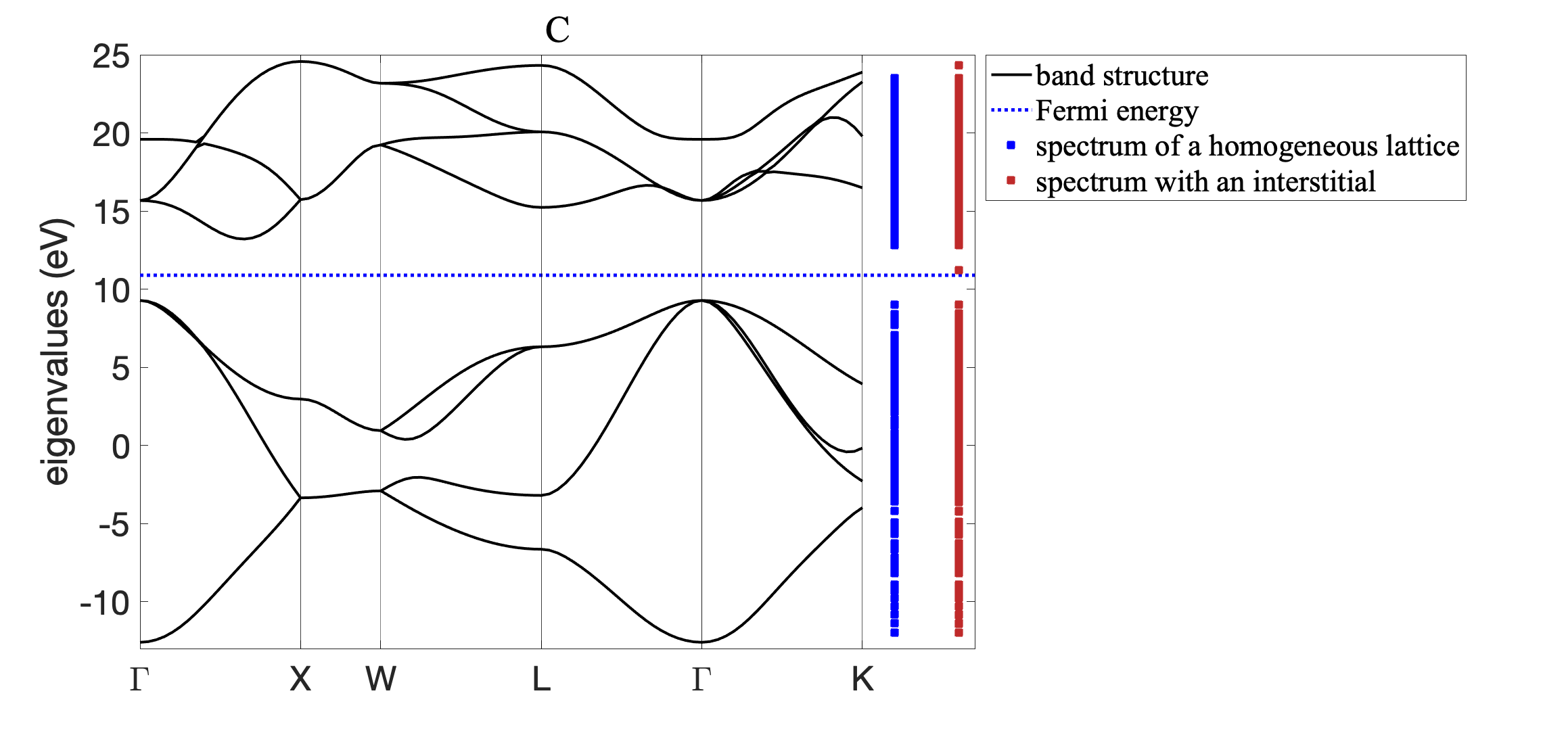}
	\caption{Band structure of C; spectrum of the homogeneous lattice (supercell approximation) and defective system.}
	\label{fig:C_band}
\bigskip
	\includegraphics[width=15cm]{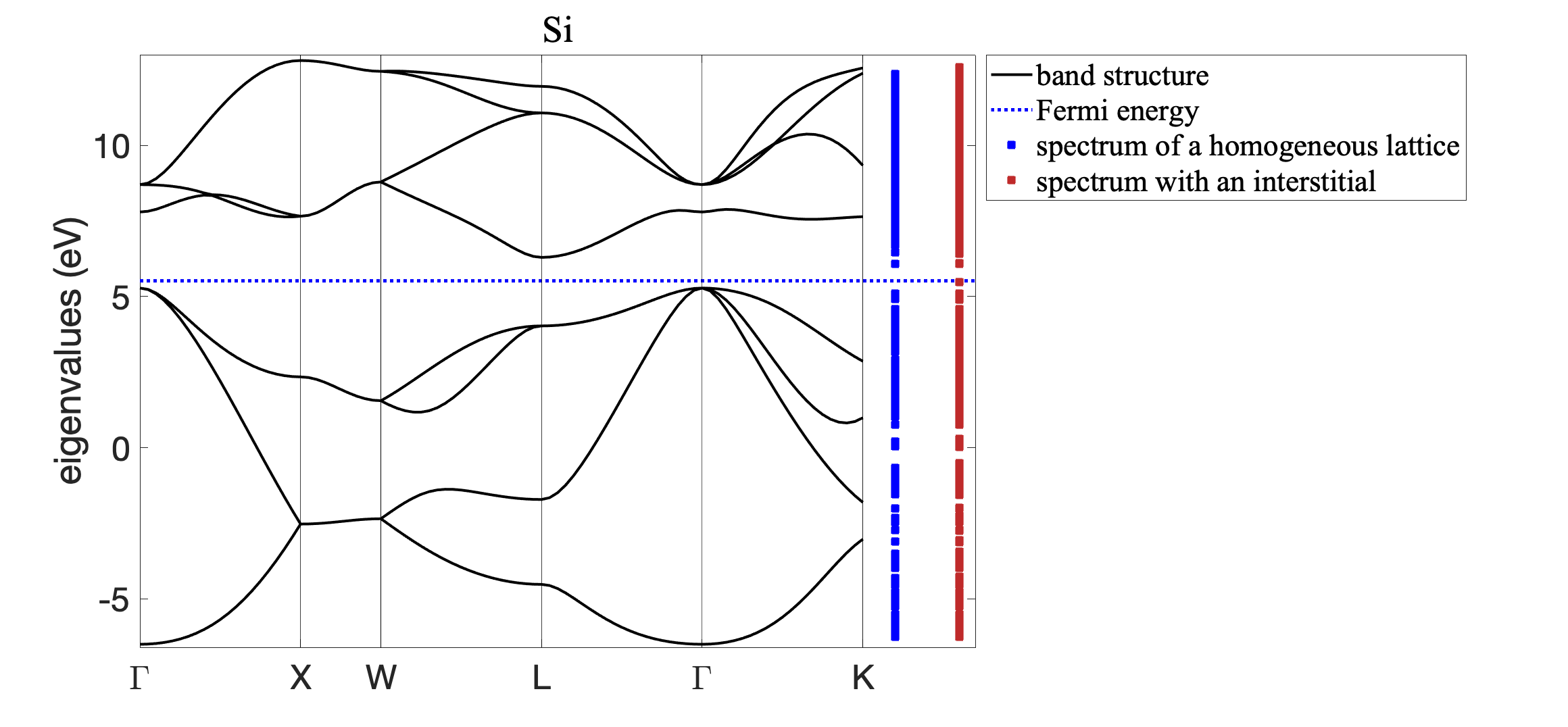}
	\caption{Band structure of Si, spectrum of the homogeneous lattice (supercell approximation) and defective system.}
	\label{fig:Si_band}
\end{figure} 

Next, we create a self-interstitial near the origin, and observe (in Figures \ref{fig:C_band} and~\ref{fig:Si_band}) the expected pollution of the band gap in the defected system. By tweaking the position of the interstitial we are able to create configurations where an eigenvalue is arbitrarily close to the Fermi-energy in order to provide a challenging situation to confirm the result of Theorem \ref{thm:improved_locality}.

\subsection{Site energy locality}
To test the locality of interatomic interaction we evaluate all first and second site energy derivatives $E_{\ell, j} = \partial_{R_j} E_\ell$ and $E_{\ell, ij} = \partial_{R_i} \partial_{R_j} E_\ell$ in both the homogeneous and defective system, and plot the data points 
\[
    \big( r_{\ell j}, |E_{\ell, j}| \big)  
    \qquad \text{and} \qquad 
    \big( r_{\ell i} + r_{\ell j}, |E_{\ell, ij}| \big)
\]
in Figures~\ref{fig:C_Es} and~\ref{fig:Si_Es}. For the homogeneous systems all sites are equivalent, hence we only plot the site energy  derivatives for a single site. For the defective systems we plot the data points for the interstitial site itself (``$|y_\ell|$ small'') as well as for the site in the computational cell that has the largest distance to the interstitial atom (``$|y_\ell|$ large'').

We clearly observe the exponential decay of interaction strength as predicted in Theorem \ref{thm:improved_locality}. 
Moreover, we also observe that for sites $\ell$ far from the defect the site derivative decay perfectly  matches that of the bulk system. 

Two additional observations were unexpected for us: (1) the decay of site derivatives for ``near-defect sites'' does not exhibit the increased prefactor that we predicted; however we do see this increase in the second derivatives. (2) the decay of interaction in the silicon system is nearly identical (after rescaling by the lattice constants) to the carbon system even though silicon has a much smaller band gap. 

These observations suggests that there are further effects leading to improved locality of interaction that our analysis does not fully capture. While a possible explanation is that the locality of the bond integral functions dominates the locality of the resolvents, this does not explain the excellent locality of the Si systems which have a fairly small band gap.

\begin{figure}[htb]
	\centering 
	\begin{subfigure}{7.7cm}
	\includegraphics[width=7.7cm]{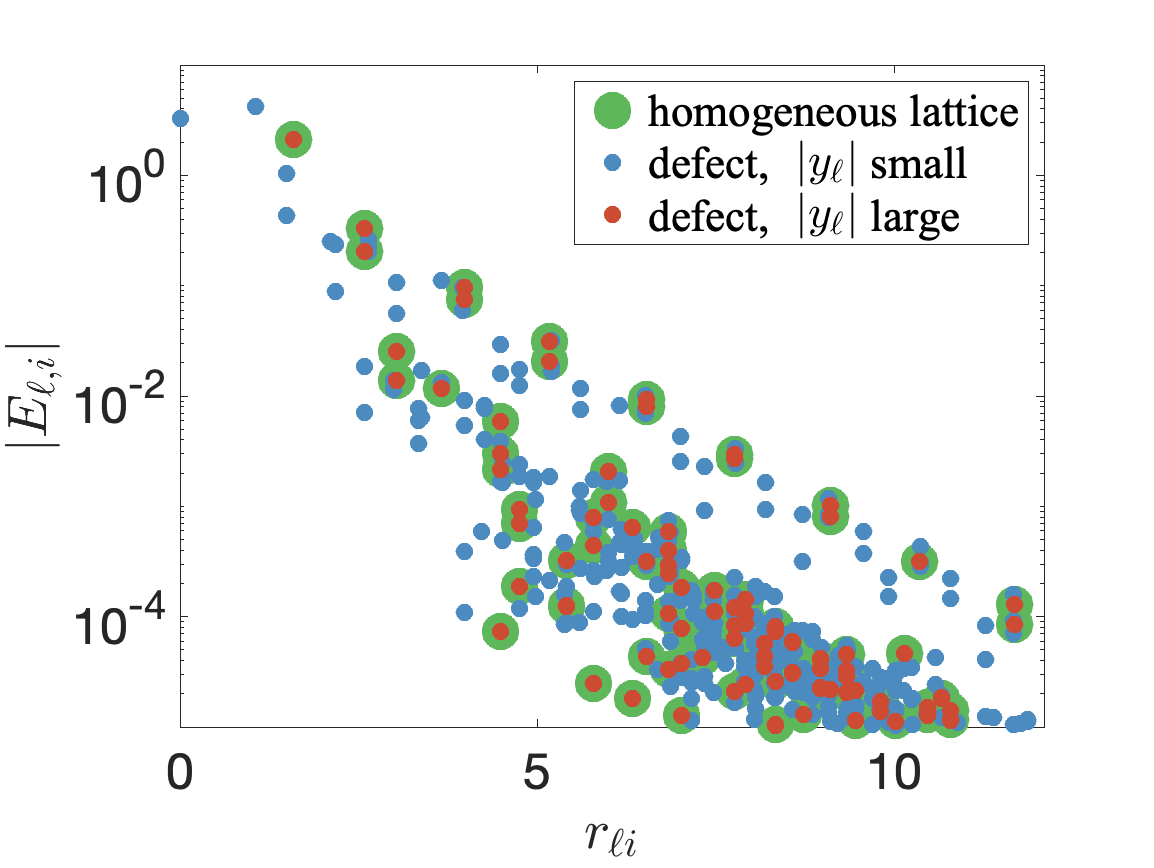}
	\caption{Decay of site energy derivatives.}
	\end{subfigure}
	\begin{subfigure}{7.7cm}
	\includegraphics[width=7.7cm]{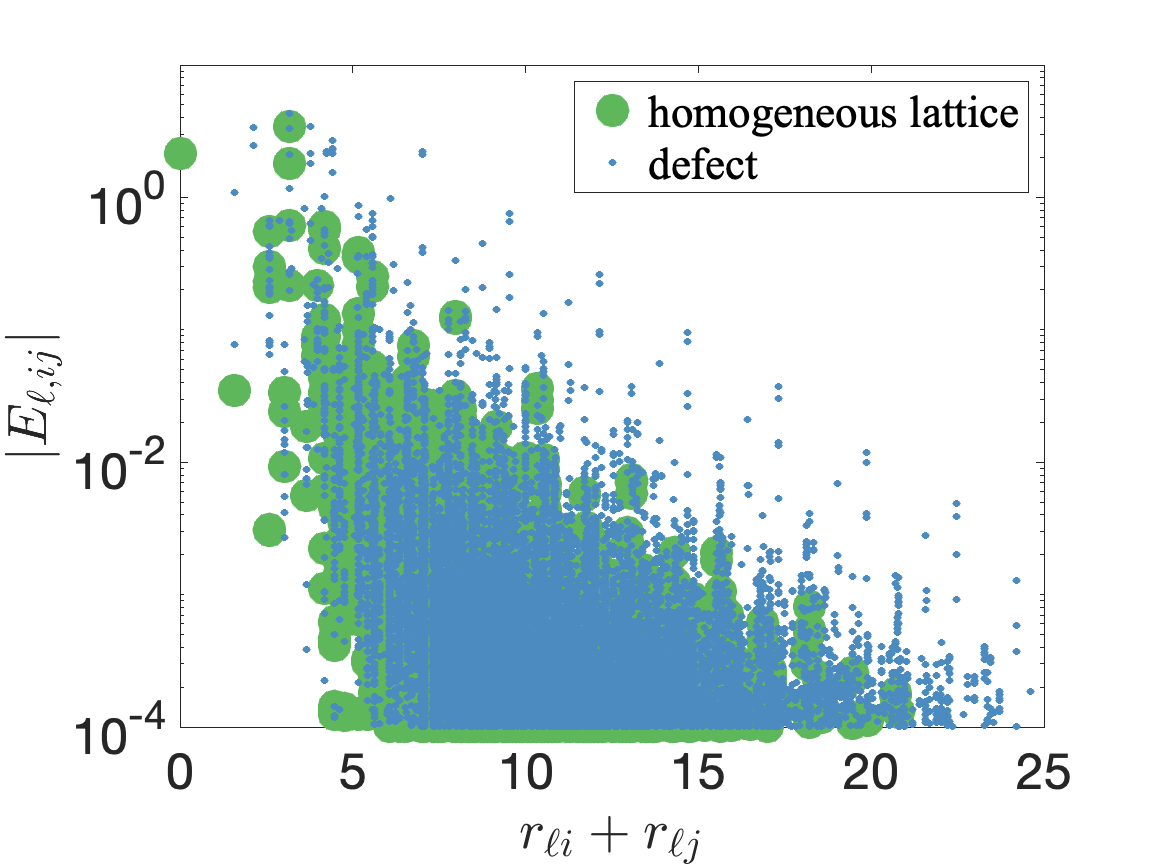}
	\caption{Decay of site energy hessians.}
	\end{subfigure}
	\caption{Carbon: Locality of site energies in homogeneous lattice and defective system.}
	\label{fig:C_Es}
\end{figure} 
\begin{figure}[htb]
	\centering 
	\begin{subfigure}{7.7cm}
	\includegraphics[width=7.7cm]{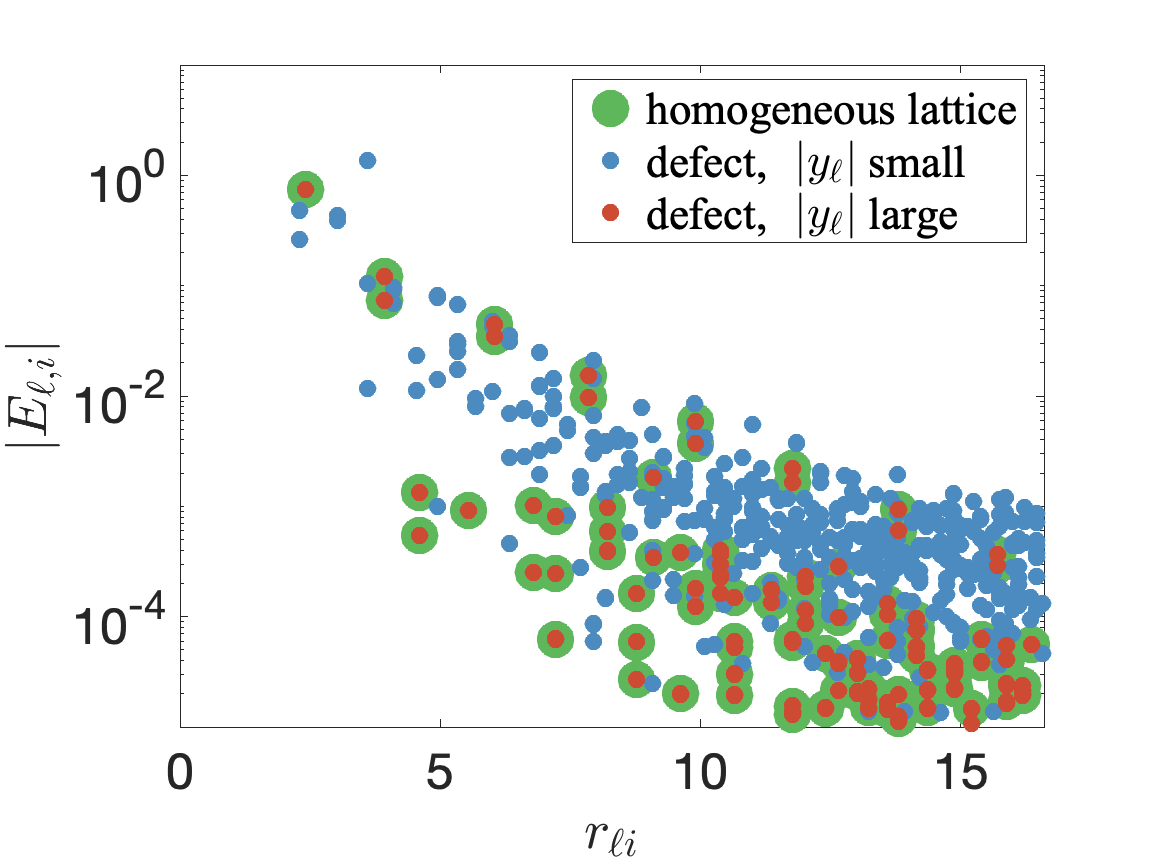}
	\caption{Decay of site energy derivatives.}
	\end{subfigure}
	\begin{subfigure}{7.7cm}
	\includegraphics[width=7.7cm]{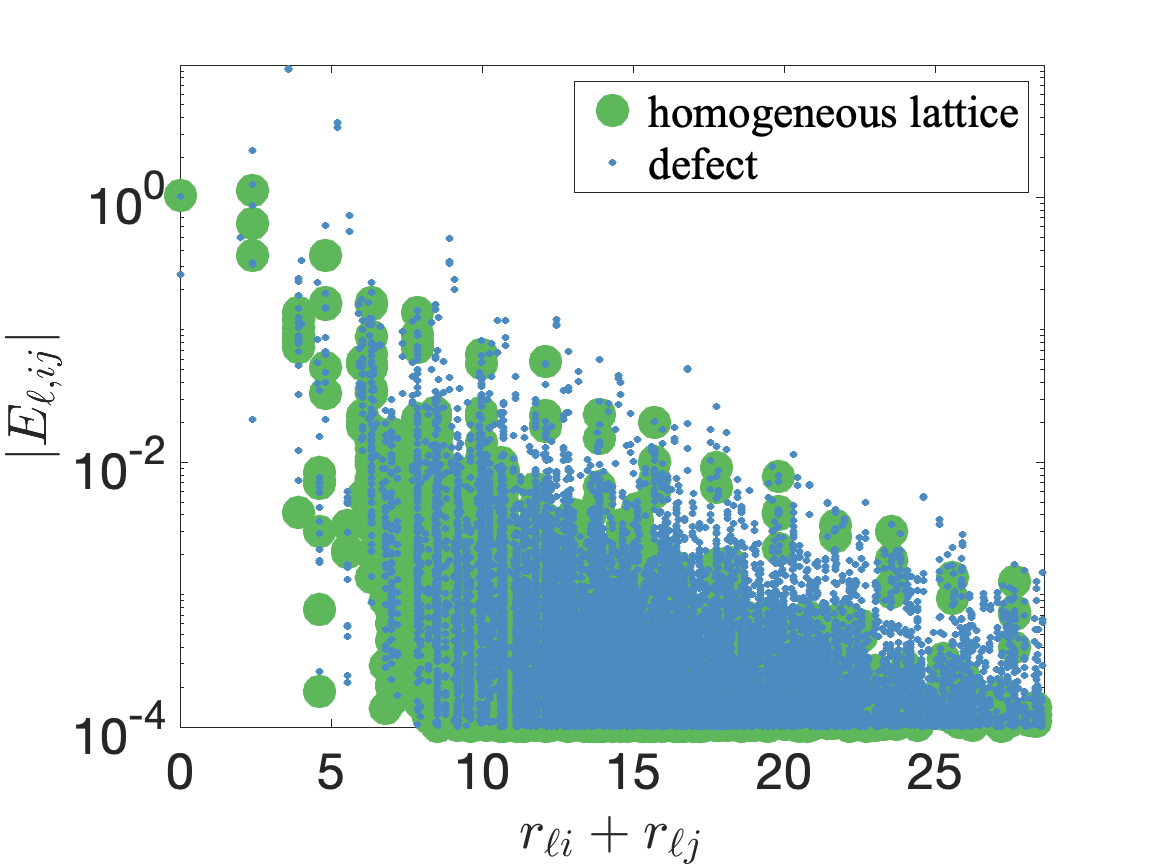}
	\caption{Decay of site energy hessians.}
	\end{subfigure}
	\caption{Silicon: Locality of site energies in homogeneous lattice and defective system.}
	\label{fig:Si_Es}
\end{figure} 

\newpage
\subsection{Force locality}
Finally, we compare the decay of site energy derivatives to  the decay of force derivatives. The reason for this additional test is that our definition of a site-energy is somewhat arbitrary. Indeed, there are infinitely many possible decompositions of total energy  into site energies and each choice may lead to a different rate of decay of the interaction. Forces, on the other hand, are uniquely defined. Their locality is therefore ``canonical'' and provides a limit for the locality of site energies.

In Figure~\ref{fig:CSi_frc}, we compare the  decay of site energy derivatives and force derivatives.
We evaluate the force derivatives $f_{\ell, j} = \partial_{R_j} f_\ell$, where the force is defined by the (negative) derivative of the total energy $f_\ell = -\partial_{R_\ell} E$, and plot the data points 
\[
    \big( r_{\ell j}, |f_{\ell, j}| \big)  
    \qquad \text{and} \qquad 
    \big( r_{\ell i} + r_{\ell j}, |E_{\ell, ij}| \big)
\]
in Figures~\ref{fig:CSi_frc}.  We observe that the site energy locality matches force locality very closely, which suggests that our choice of site energies leads to near-optimal locality of interaction.

\begin{figure}[htb]
	\centering 
	\begin{subfigure}{7.7cm}
	\includegraphics[width=7.7cm]{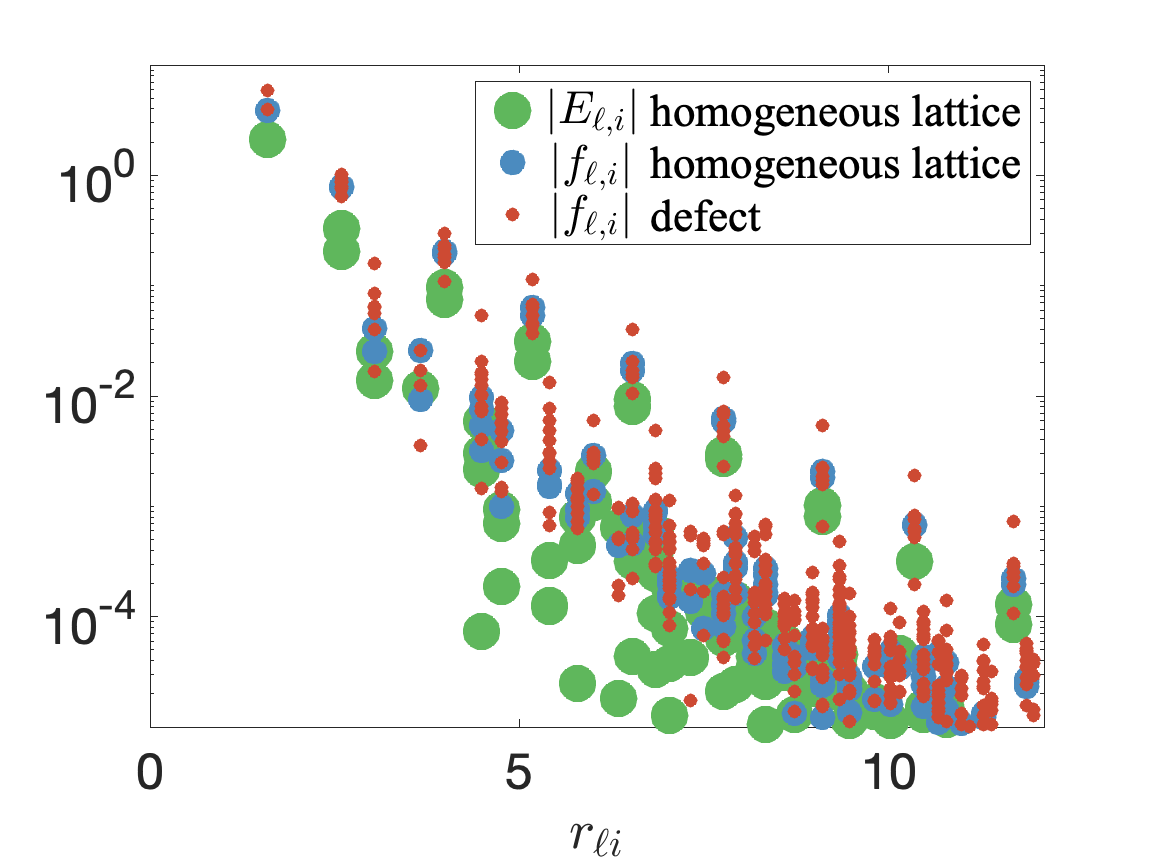}
	\caption{Carbon.}
	\end{subfigure}
	\begin{subfigure}{7.7cm}
	\includegraphics[width=7.7cm]{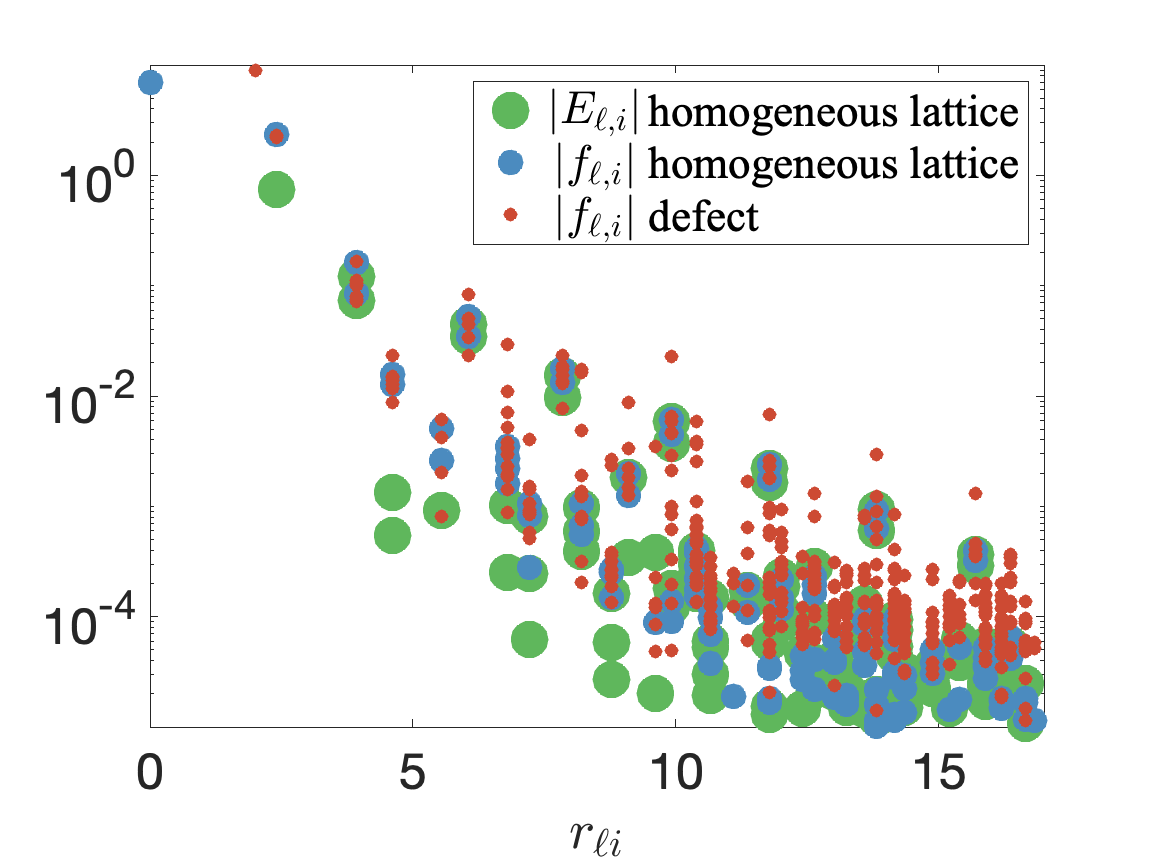}
	\caption{Silicon.}
	\end{subfigure}
	\caption{The decay of force derivatives in homogeneous lattice and defective system.}
	\label{fig:CSi_frc}
\end{figure} 





\section{Conclusions}\label{sec:conclusions}
We have extended the results of \cite{ChenOrtner16} to the zero Fermi-temperature case under the assumption that the chemical potential is not an eigenvalue of the Hamiltonian. We have described a site energy decomposition for a zero Fermi-temperature linear tight binding model and shown that the site contributions are exponentially localised. Most importantly, we have shown that the exponents in these estimates are independent of the discrete spectrum inside the band gap caused by point defects, and even the pre-factors converge to the pre-factors that would result from using the homogeneous site energy in the estimates, as the distance of a site to the defect increases. Our numerical results in \S~\ref{sec:numerics} strongly support our analysis, but also point to possible further extensions in particular in the limit of small band gaps where our results may not yet be sharp.

The same analysis was also applied to the Helmholtz free energy in the canonical ensemble under the assumption that the Fermi level is fixed. In particular, this improves the locality results of \cite{ChenOrtner16} for insulators. Moreover, the analysis carries over to other quantities of interest as in \cite{ChenLuOrtner18}. 

The results of this paper allow us to formulate zero Fermi-temperature lattice relaxation as a variational problem on the energy space of displacements. In particular, for $y\in\mathrm{Adm}(\Lambda)$, we can define
\begin{align*} \mathcal G(y) \coloneqq \sum_{\ell }\left( G_{\ell }(y) - G_{\ell}(x) \right). \end{align*}
This {\em grand potential difference functional} is well defined if $\mu \not \in \sigma(\Ham(y))$. This can be shown by applying results of \cite{ChenNazarOrtner19} together with the locality estimates of this paper. We can then consider the meaningful problem: for fixed $\delta > 0$,
\begin{align}\label{eq:minproblem} 
\overline{y}\in\arg\min\{ \mathcal G(y) \colon y \in \mathrm{Adm}(\Lambda), \, (\mu - \delta, \mu + \delta )\cap \sigma(\Ham(y)) = \emptyset\} 
\end{align}
where ``$\arg\min$'' denotes the set of local minimisers. The locality results presented in this paper allow us to show that the site energies and their derivatives converge exponentially quickly in the zero Fermi-temperature and thermodynamic limits. This observation allows us to prove that \cref{eq:minproblem} is the limiting model of analogous finite Fermi-temperature and finite domain size models. Rigorous results are presented in a forthcoming paper \cite{inprep}.

\section{Proofs of the Main Results}\label{sec:proofs}
\subsection{Definition of the Site Energies}
\label{proofs:definition_site_energies}

Before we begin the proof of the locality estimates, we need to show that the definition of the finite Fermi-temperature site energy is valid. That is, we need $\mathfrak{g}^\beta(\,\cdot\,;\mu)$ to extend to a holomorphic function on some open neighbourhood of the spectrum and we need to consider appropriate contours $\mathscr C_\beta$, $\mathscr C_\infty$. 
\begin{lemma}[Analytic Continuation of $\mathfrak{g}^\beta(z;\mu)$]
	\label{lem:analytic-cont}
	Fix $\beta \in(0,\infty)$. Then, $z\mapsto\mathfrak{g}^\beta(z;\mu)$ can be analytically continued to the set $\mathbb{C} \setminus \{ \mu + i r  \colon r \in \mathbb R, |r| \geq \pi \beta^{-1} \}$.
\end{lemma}
\begin{proof}
	Extending $\mathfrak{g}^\beta(\,\cdot\,;\mu)$ into the complex plane amounts to choosing a branch cut of the complex logarithm. For each $n\in\mathbb N$, we define,
	\begin{gather}\label{eq:analytic_cont}\begin{split} \mathfrak{g}_n^\beta(z;\mu) \coloneqq \frac{2}{\beta}\big[ \log\left|1-f_\beta(z-\mu)\right| + i\textrm{Arg}_n( 1 - f_\beta(z-\mu)) \big]\quad \text{where}\\
	\textrm{Arg}_n(z) = \textrm{Arg}(z) \quad (\text{mod }2\pi)\quad\text{ and } \quad \textrm{Arg}_n(z) \in ((n-1)\pi, (n+1)\pi].
	\end{split}\end{gather}
	Choosing the principal branch of the complex logarithm, we get $\mathfrak{g}^\beta_0$ which agrees with $\mathfrak{g}^\beta$ on the real axis.
	
	To simplify notation, and without loss of generality, we suppose $\mu = 0$. 
	
	It is clear that the Fermi-Dirac distribution has isolated singularities at $i(2k+1)\pi\beta^{-1}$ for all $k \in \mathbb Z$ (that is, when $e^{\beta z} = -1$) and is holomorphic away from these singularities. Therefore, $\mathfrak{g}_n^\beta(\,\cdot\,;0)$ is holomorphic on the set that avoids the branch cut of the complex logarithm and the non-analyticity of $1-f_\beta$. That is, $\mathfrak{g}_n^\beta(\,\cdot\,;0)$ is holomorphic on
	\begin{equation*}\begin{split} \begin{cases}\{ z \in \mathbb C \colon 1 - f_\beta(z) \not\in (-\infty,0]  \} \setminus \left\{ \tfrac{(2k + 1)\pi i}{\beta}\right\}_{k \in \mathbb Z} &\text{for }n \text{ even}, \\
	\{ z \in \mathbb C \colon 1 - f_\beta(z) \not\in [0,\infty)  \} \setminus \left\{ \tfrac{(2k + 1)\pi i}{\beta}\right\}_{k \in \mathbb Z} &\text{for }n \text{ odd}. \end{cases} \end{split}\end{equation*}
	Rewriting $1-f_\beta$ we obtain,
	\begin{align}\label{eq:1-f}
	1 - f_\beta(z) = \frac{e^{\beta z}}{1 + e^{\beta z}} = \frac{e^{\beta z} \left( 1+e^{\beta\overline{z}} \right)}{|1+e^{\beta z}|^{2}} = \frac{e^{\beta\Re{z}}}{|1+e^{\beta z}|^2}\left( e^{i\beta \Im{z}} + e^{\beta \text{Re}(z)} \right).
	\end{align}
	The factor, $e^{\beta\textrm{Re}(z)}|1+e^{\beta z}|^{-2}$, is real and positive and so $1-f_\beta(z)$ avoids the branch cut if and only if $h(z) \coloneqq e^{i\beta \Im{z}} + e^{\beta \Re{z}}$ does. Now, $h(z) \in (-\infty, 0]$ if and only if $\Re{z}\leq0$ and $\beta\Im{z} = (2k+1)\pi$ for some $k \in \mathbb Z$. On the other hand $h(z) \in [0,\infty)$ if and only if $\beta \Im{z} = 2k\pi$ for some $k \in \mathbb Z$. 
	We can therefore conclude that $\mathfrak{g}^\beta_0(\,\cdot\,;0)$ is holomorphic on 
	\[ 
	A_\beta^0 \coloneqq \Big\{ z \in \mathbb C \colon \Re{z}>0 \Big\} \cup \Big\{ z \in \mathbb C \colon \beta\Im{z} \in (-\pi, \pi) \Big\}
	\]
	and that $\mathfrak{g}^\beta_n(\,\cdot\,;0)$ (for $n \not= 0$) is holomorphic on the set
	\[ 
	A_\beta^n \coloneqq \Big\{ z \in \mathbb C \colon \Re{z} < 0,\,\beta\Im{z} \in ((n-1)\pi, (n+1)\pi) \Big\}. 
	\]

	Since $A_\beta^n \cap A_\beta^{n+1} = \{z\in\mathbb C\colon \Re{z}<0,\,\beta \Im{z} \in (n\pi, (n+1)\pi )\} $, and by \cref{eq:1-f}, we have that \[\textrm{Arg}_n(1 - f_\beta(z)) = \textrm{Arg}_{n+1}(1 - f_\beta(z)) \in ( n\pi, (n+1)\pi ]\]for all $z \in A_\beta^n \cap A_\beta^{n+1}$. That is, $\mathfrak{g}^\beta_n(\,\cdot\,;0) = \mathfrak{g}^\beta_{n+1}(\,\cdot\,;0)$ on 
	$A_\beta^n \cap A_\beta^{n+1}$. 
	
	We may therefore consider the analytic continuation of $\mathfrak{g}^\beta_n(\,\cdot\,;0)$ to $A_\beta^n \cup A_\beta^{n+1}$. We do this for each $n\in \mathbb Z$ and conclude since $\bigcup_{n\in\mathbb Z} A_\beta^n = \mathbb C \setminus \{ i r \colon |r| \geq \pi \beta^{-1} \}$.
\end{proof}
From now on, we denote the analytic continuation by $z\mapsto\mathfrak{g}^\beta(z;\mu)$. We need conditions on the family of contours, $\{ \mathscr C_\beta \}_\beta$, to ensure that $\mathfrak{g}^\beta(z;\mu)$ remains uniformly bounded for $z \in \mathscr C_\beta$ and $\beta > 0$. We suppose there exists some $\beta$-independent constant $0 < \ctAwayFromSingularity < \pi$ such that
\begin{equation}\label{eq:distance_singularity_beta}\begin{aligned}
&\textrm{dist}\left(z, \{ \mu \pm i\pi\beta^{-1} \}\right) \geq \ctAwayFromSingularity\beta^{-1} \quad  \forall z \in \mathscr C_\beta\quad\text{and}\\
&\qquad\qquad\exists A \subset \mathbb C \text{ bounded s.t. } \mathscr C_\beta \subset A
\end{aligned}\end{equation}for all $\beta>0$.
\begin{lemma}\label{lem:boundedness-g}
	Fix $y \in \mathrm{Adm}(\Lambda)$. Suppose that $\{\mathscr C_\beta\}_\beta$ is a family of simple closed contours encircling $\sigma(\Ham(y))$ and satisfying \cref{eq:distance_singularity_beta}. Then, for $\beta_0 > 0$, \[\sup_{\beta\geq\beta_0}\sup_{z \in \mathscr C_\beta}|\mathfrak{g}^\beta(z;\mu)| < \infty.\]
\end{lemma}
\begin{proof}
	Since $A$ is bounded, we can find a strip of width $r>0$ about the real axis containing $A$. This means that for fixed $\beta > 0$, the number branches of the complex logarithm that we must consider, as in \cref{eq:analytic_cont}, in order to have extended $\mathfrak{g}^\beta(\,\cdot\,;\mu)$ to the whole of $A$ is at most a constant multiple of $\tfrac{r\beta}{\pi}$. This means that
	\[ |\mathrm{Arg}_n(1 - f_\beta(z-\mu))| \leq (n+1)\pi \leq Cr \beta \]
	for all $n$ such that $A\cap\{ z \in \mathbb C \colon \beta \Im{z} \in ((n-1)\pi,(n+1)\pi) \} \not=\emptyset$ and $z \in A_\beta^n$. Therefore, for $z \in A$ and $\beta > 0$, we have that $\Im{\mathfrak{g}^\beta(z;\mu)}$ is bounded on $A$ independently of $\beta$.
	
	Fix $\beta > \beta_0$. Now we show that, away from the singularities, $\Re{\mathfrak{g}^\beta(\,\cdot\,;\mu)}$ is uniformly bounded. We know that $\Re{\mathfrak{g}^\beta(z;\mu)} = 2\Re{z-\mu} - \tfrac{2}{\beta}\log\left|1 + e^{\beta (z-\mu)}\right|$ and $2(z-\mu)$ is uniformly bounded on $A$. Moreover, 
	\begin{equation*}
	\tfrac{2}{\beta}\log\left|1 + e^{\beta (z-\mu)}\right| \leq \tfrac{2}{\beta}\log\left(1 + \exp\left(\beta\sup_{z \in A} |\Re{z} - \mu|\right)\right) \leq C 
	\end{equation*}for some $C>0$ depending only on $A$ and $\beta_0$. Therefore, all that is left to show is that $|1 + e^{\beta(z-\mu)}|$ is uniformly bounded below by a positive constant. If $\Re{z - \mu} < -c\beta^{-1}$ for some $c>0$ then $|1 + e^{\beta(z-\mu)}| \geq 1 - e^{-c} > 0$ and if $\Re{z - \mu} > c\beta^{-1}$ then $|1 + e^{\beta(z-\mu)}| \geq e^{c} - 1 > 0$. On the other hand, if $|\beta\Im{z - \mu} - r|\geq \theta$ for all $r \in \mathbb R$ such that $|r|\geq \pi$, then $|1+e^{\beta(z-\mu)}|\geq \tan(\theta)>0$.
\end{proof}

\subsection{Proof of Propositions \ref{prop:existing_locality} and \ref{prop:locality_insulators}: Locality Estimates}\label{proofs:existing_locality}
We now briefly sketch the main ideas in the proof of Propositions~\ref{prop:existing_locality} and \ref{prop:locality_insulators}. We mainly do this so that we can track the $\beta$-dependent constants in the proof. A key ingredient is the following Combes-Thomas type  estimate on the resolvent:
\begin{lemma}[Combes-Thomas]\label{lem:Thomas-Combes}
	Fix $y \in \textup{Adm}(\Lambda)$ and $z \in \mathbb C$ such that $\emph{\textrm{dist}}\left( z, \sigma(\Ham(y)) \right) \geq \mathfrak{d}$ for some $\mathfrak{d}>0$. Then, 
	\[ \left|\left(\Ham(y) - z\right)^{-1}_{\ell k,ab}\right| \leq \tfrac{2}{\mathfrak{d}} e^{-\ctCT(\mathfrak{d}) |y(\ell) - y(k)|}, \]
	where $\ctCT(\mathfrak{d}) \coloneqq  c \min\{1,\mathfrak{d}\}$ for some $c>0$ depending only on $\ctTBprefactor{0},\ctTBexponent{0},\ctnoninterpen$ and $d$.
\end{lemma}
\begin{proof}This follows the proof of \cite[Lemma~6]{ChenOrtner16} and the main ideas of \cite{ELu10}. The claimed $\mathfrak{d}$ dependence in the exponent can be obtained by replacing (34) in \cite[Lemma~6]{ChenOrtner16} with the following sharper estimate: there exists a $C>0$ such that
\begin{equation}\label{eq:CT-proof}\sup_{\ell \in \Lambda} \sum_{k \in \Lambda} \left|\left[\Ham(y)\right]_{\ell k}\right|\left(e^{\ctCT|y(\ell) - y(k)|} - 1\right) \leq C \ctCT \end{equation}
for all $0 \leq \ctCT \leq \tfrac{1}{2}\ctTBexponent{0}$. To conclude, we note that in the proof of \cite[Lemma~6]{ChenOrtner16}, $\ctCT>0$ must be chosen sufficiently small such that the right hand side of \cref{eq:CT-proof} is less than $\tfrac{1}{2}\mathfrak{d}$.
\end{proof}

To simplify notation, we shall write $r_{\ell k}(y) \coloneqq |y(\ell) - y(k)|$ for the distance between two atomic sites and $\mathscr R_z(y) \coloneqq (\Ham(y) - z)^{-1}$ for the resolvent operator corresponding to $\Ham(y)$. We will drop the argument $(y)$ in $r_{\ell k}(y)$, the resolvent and Hamiltonian when the dependence on $y$ is clear from context. Moreover, we shall use the following shorthand for derivatives of the Hamiltonian: for $\bm{m} = (m_1,\dots,m_j) \in \Lambda^j$ and $\bm{i} = (i_1,\dots,i_j) \in \left(\mathbb R^d\right)^j$, define
\[ \left[\Ham_{,\bm{m}}\right]_{\bm{i}} \coloneqq \frac{\partial^j \Ham(y)}{\partial [y(m_1)]_{i_1} \dots \partial [y(m_j)]_{i_j}}.\]Often, to simplify notation further, we will drop the Euclidean coordinate in this notation.  

Before we prove the locality estimates of Propositions~\ref{prop:existing_locality} and \ref{prop:locality_insulators}, we shall derive general bounds for the first and second derivatives of the resolvent. We choose $z\in\mathbb C$ such that
\begin{equation} 
{\mathrm{dist}}\left( z, \sigma(\Ham(y)) \right) \geq \mathfrak{d}
\end{equation}
for some $\mathfrak{d}>0$. Moreover, we let $\ctCT = \ctCT(\mathfrak{d})>0$ be the corresponding Combes-Thomas exponent from Lemma~\ref{lem:Thomas-Combes} and define $\mathfrak{d}_j \coloneqq \min\{\ctTBexponent{1},\dots,\ctTBexponent{j},\ctCT\}$ for each $j$. By applying the Combes-Thomas estimate and using the regularity of the Hamiltonian, we have
\begin{equation}\label{eq:first_derivative_resolvent} \begin{split}
\left|\frac{\partial \left[\mathscr R_z(y)\right]^{aa}_{\ell\ell}}{\partial [y(m)]_i}\right| 
&= \left| \sum_{\above{\ell_1,\ell_2\in\Lambda}{1\leq b,c \leq N_\mathrm{b}}} \left[\mathscr R_z\right]_{\ell\ell_1}^{ab} \big([\Ham_{,m}]_i\big)_{\ell_1\ell_2}^{bc}\left[\mathscr R_z\right]_{\ell_2 \ell}^{ca} \right|  \\
&\leq 4 N_\mathrm{b}^2 \ctTBprefactor{1} \mathfrak{d}^{-2} \sum_{\ell_1, \ell_2 \in \Lambda} e^{-\ctCT r_{\ell_1 \ell}} e^{-\ctTBexponent{1} \left( r_{\ell_1m} + r_{m\ell_2}\right)} e^{- \ctCT r_{\ell_2 \ell}} \\&\leq C \mathfrak{d}^{-2} \left(\sum_{\ell_1\in \Lambda} e^{-\mathfrak{d}_1 \left(r_{\ell \ell_1} + r_{\ell_1m}\right)} \right)^2 \\&\leq C \mathfrak{d}^{-2} \mathfrak{d}_1^{-2d} e^{-\mathfrak{d}_1 r_{\ell m}}.
\end{split}\end{equation}
Similarly, for the second derivatives,
\begin{equation}\label{eq:second_derivatives_resolvent}\begin{split}
\frac{\partial^2 \left[\mathscr R_z(y)\right]^{aa}_{\ell\ell}}{\partial y(m_1) \partial y(m_2)} = &\bigg[\mathscr R_z \Ham_{,m_1} \mathscr R_z \Ham_{,m_2} \mathscr R_z - \mathscr R_z \Ham_{,m_1m_2} \mathscr R_z + \mathscr R_z \Ham_{m_2} \mathscr R_z \Ham_{m_1} \mathscr R_z\bigg]^{aa}_{\ell\ell}.
\end{split}
\end{equation} 
Each of the terms in (\ref{eq:second_derivatives_resolvent}) can be bounded separately:
\begin{equation*}\begin{split}
\big|\big[\mathscr R_z &\Ham_{,m_1} \mathscr R_z \Ham_{,m_2} \mathscr R_z\big]_{\ell\ell}^{aa}\big| \\&\leq 8 N_\textrm{b}^4 \ctTBprefactor{1}^2 \mathfrak{d}^{-3} \sum_{\ell_1,\ell_2,\ell_3,\ell_4 \in \Lambda} e^{-\ctCT \left( r_{\ell\ell_1} + r_{\ell_2\ell_3} + r_{\ell_4\ell}\right)} e^{-\ctTBexponent{1} \left( r_{\ell_1m_1} + r_{m_1\ell_2} + r_{\ell_3m_2} + r_{m_2\ell_4}\right)} \\
&\leq C \mathfrak{d}^{-3} \mathfrak{d}_1^{-4d} e^{-\frac{1}{2}\mathfrak{d}_1\left(r_{\ell m_1} + r_{\ell m_2}\right) }; \quad \text{and}\\
\big|\big[\mathscr R_z& \Ham_{,m_1m_2} \mathscr R_z\big]_{\ell\ell}^{aa}\big| \\&\leq 4 N_\textrm{b}^2 \ctTBprefactor{2} \mathfrak{d}^{-2} \sum_{\ell_1,\ell_2 \in \Lambda} e^{-\ctCT \left( r_{\ell\ell_1} + r_{\ell_2\ell}\right)} e^{-\ctTBexponent{2} \left( r_{\ell_1m_1} + r_{\ell_1m_2}  + r_{\ell_2m_1}  r_{\ell_2m_2} \right)} \\
&\leq C\mathfrak{d}^{-2} \left(\sum_{\ell_1 \in \Lambda} e^{-\ctCT r_{\ell\ell_1} } e^{-\ctTBexponent{2} \left( r_{\ell_1m_1} + r_{\ell_1m_2} \right)}\right)^2 
\leq C\mathfrak{d}^{-2} \mathfrak{d}_2^{-2d} e^{-\frac{1}{2}\mathfrak{d}_2\left(r_{\ell m_1} + r_{\ell m_2}\right)}.\end{split}
\end{equation*}
Therefore, we obtain the following bound:
\begin{equation}\label{eq:second_derivative_resolvent_estimate}\begin{split}
\left| \frac{\partial^2 \left[\mathscr R_z(y)\right]^{aa}_{\ell\ell}}{\partial y(m_1) \partial y(m_2)} \right| \leq C \mathfrak{d}^{-3} \max\Big\{ \mathfrak{d}_1^{-4d}, \mathfrak{d} \mathfrak{d}_2^{-2d} \Big\} \, e^{-\frac{1}{2}\mathfrak{d}_2\left(r_{\ell m_1} + r_{\ell m_2}\right) }.\end{split}
\end{equation}
In particular, for $\mathfrak{d}$ sufficiently small, we have
\begin{equation}\label{eq:resolvent_estimates_d_small}
\left| \frac{\partial \left[\mathscr R_z(y)\right]^{aa}_{\ell\ell}}{\partial y(m)} \right| \lesssim \mathfrak{d}^{-2(d +1)}\, e^{-\ctCT r_{\ell m}} \quad \text{and} \quad \left| \frac{\partial^2 \left[\mathscr R_z(y)\right]^{aa}_{\ell\ell}}{\partial y(m_1) \partial y(m_2)} \right| \lesssim \mathfrak{d}^{-(4d +3)}\, e^{-\frac{1}{2}\ctCT\left(r_{\ell m_1} + r_{\ell m_2}\right) }
\end{equation}
It should be clear that, for higher derivatives, the same arguments can be made and similar estimates hold.

\begin{proof}[Proof of Proposition~\ref{prop:existing_locality}:  Finite Temperature Locality for Metals]

	We will only consider the case where $j\in\{1,2\}$. For $j> 2$, similar arguments can be made but is omitted as the notation becomes tedious and no new ideas are used. Since for all $z \in \mathscr C_\beta$,
	\[ 
	\textrm{dist}\Big(z,\sigma(\Ham(y))\Big) \geq \frac{\pi}{2\beta}, 
	\]
	we may use (\ref{eq:first_derivative_resolvent}) and (\ref{eq:second_derivative_resolvent_estimate}) with $\ctCT = \ctCT(\tfrac{\pi}{2\beta})$. First, we consider $j=1$ and write:
	\begin{align*}
	\left|\frac{\partial G_{\ell}^\beta(y)}{\partial [y(m)]_i}\right| &\leq \frac{1}{2\pi}\sum_a \left| \oint_{\mathscr C_\beta} \mathfrak{g}^\beta(z;\mu) \frac{\partial \left[\mathscr R_z(y)\right]_{\ell\ell}}{\partial [y(m)]_i} \rm{d}z  \right| \\
	&\leq C\numorbitals |\mathscr C_\beta|\max_{\mathscr C_\beta} \left| \mathfrak{g}^\beta(z;\mu) \right|\beta^2 \min\left\{\ctTBexponent{1}, \ctCT(\tfrac{\pi}{2\beta})\right\}^{-2d} e^{-\textrm{min}\{\ctTBexponent{1}, \ctCT(\frac{\pi}{2\beta})\}\, r_{\ell m} }
	\end{align*}
	By \cref{eq:distance_singularity_beta} and Lemma~\ref{lem:boundedness-g}, $\mathfrak{g}^\beta(\,\cdot\,;\mu)$ is uniformly bounded along $\mathscr C_\beta$ independently of $\beta$. We can thus conclude with $\eta_1 \coloneqq  \textrm{min}\{\ctTBexponent{1}, \ctCT(\frac{\pi}{2\beta})\}$.
	
	Similarly, for $j = 2$, we may apply (\ref{eq:second_derivative_resolvent_estimate}) together with
	\begin{align*}
	\left|\frac{\partial^2 G_{\ell}^\beta(y)}{\partial [y(m_1)]_{i_1}\partial [y(m_2)]_{i_2}}\right|	&\leq C\numorbitals |\mathscr C_\beta| \sup_{z \in \mathscr C_\beta}\left| \frac{\partial^2 \left[\mathscr R_z(y)\right]^{aa}_{\ell\ell}}{\partial y(m_1) \partial y(m_2)}\right|,
	\end{align*}
	to conclude with $\eta_2 \coloneqq  \frac{1}{2}\textrm{min}\{\ctTBexponent{1}, \ctTBexponent{2}, \ctCT(\tfrac{\pi}{2\beta})\}$. The fact that, for sufficiently large $\beta>0$, the pre-factor is $C\beta^\alpha$ for some $\alpha = \alpha(j,d)>0$ should be clear from \cref{eq:resolvent_estimates_d_small} and Lemma~\ref{lem:Thomas-Combes}.
\end{proof}
In the case of an insulator, the separation between the spectrum and the contour can be chosen to be $\beta$-independent and equal to $\ctgapfinitetemp(y)$ as in \cref{eq:insulator_distance}. In the zero temperature case, this constant may be chosen to be $\tfrac{1}{2}\ctgap(y)$ where $\ctgap(y)$ is the constant from \cref{eq:insulator_distance_2}. 

\begin{proof}[Proof of Proposition~\ref{prop:locality_insulators}: Locality Estimates for Insulators]
	The proof follows in the exact same way as Proposition~\ref{prop:existing_locality} with $\ctCT=\ctCT(\ctgapfinitetemp(y))$ for finite Fermi-temperature. In the case of zero Fermi-temperature, we use the proof of Proposition~\ref{prop:existing_locality} with $\ctCT=\ctCT(\tfrac{1}{2}\ctgap(y))$ and the fact that $\left|2(z - \mu)\right|$ is uniformly bounded along the contour $\mathscr C_\infty$.
\end{proof}

\subsection{Decomposition of the Spectrum}
\label{proofs:decomposition_spectrum}
We need to show that the defective Hamiltonian can be written in terms of the reference Hamiltonian. However, we are considering a point defect reference configuration, $\Lambda$, for which $\Lambda \cap B_\mathrm{def} \not= \Lambda^\mathrm{ref} \cap B_\mathrm{def}$ in general. This means that the defective and reference Hamiltonians may be defined on different spaces. We shall extend the definitions to $\Lambda \cup \Lambda^\mathrm{ref}$: for $y \in \textrm{Adm}(\Lambda)$ and $\ell, k \in \Lambda \cup \Lambda^\mathrm{ref}$, let us define 
\begin{align}\label{eq:extended_ham}
\widetilde{\Ham}(y)_{\ell k}^{ab} \coloneqq \begin{cases}
\Ham(y)^{ab}_{\ell k} &\text{if } \ell, k \in \Lambda \\
0 &\text{otherwise}
\end{cases} \quad \text{and}\quad [\widetilde{\Ham}^\mathrm{ref}]_{\ell k}^{ab} \coloneqq
\begin{cases}
[\Ham^\mathrm{ref}]_{\ell k}^{ab} &\text{if } \ell, k \in \Lambda^\mathrm{ref} \\
0 &\text{otherwise}.
\end{cases}
\end{align}
This only changes the spectrum by introducing additional zero eigenvalues. We shall shift the spectrum away from $\{0\}$ so that we can replace $\Ham(y)$ and $\Ham^\mathrm{ref}$ with $\widetilde{\Ham}(y)$ and $\widetilde{\Ham}^\mathrm{ref}$, respectively. This does not lead to any problems as we will now see: fix $\beta \in (0,\infty]$ and an appropriate contour $\mathscr C$. Choosing $z_0 \in\mathbb C$ such that the contour $\mathscr C + z_0$ does not encircle $\{0\}$, we have 
\begin{equation*}\begin{split} G^\beta_{\ell}(y) &= - \frac{1}{2\pi i}\sum_a\oint_{\mathscr C} \mathfrak{g}^\beta(z;\mu) \left[\mathscr R_z(y)\right]_{\ell \ell}^{aa}\textrm{d}z \\
&= - \frac{1}{2\pi i}\sum_a\oint_{\mathscr C + z_0} \mathfrak{g}^\beta(z-z_0;\mu) \Big[\big(\Ham(y) -(z-z_0)\big)^{-1}\Big]_{\ell \ell}^{aa}\textrm{d}z \\
&= - \frac{1}{2\pi i}\sum_a\oint_{\mathscr C + z_0} \mathfrak{g}^\beta(z;\mu+z_0) \left[\left((\Ham(y) + z_0)^{\Lambda\cup\Lambda^\mathrm{ref}} - z\right)^{-1}\right]_{\ell \ell}^{aa}\textrm{d}z
\end{split}\end{equation*}
where $(\Ham(y) + z_0)^{\Lambda\cup\Lambda^\mathrm{ref}}$ is the extension of $\Ham(y) + z_0$ to $\Lambda\cup \Lambda^\mathrm{ref}$ as in \cref{eq:extended_ham}. Therefore, by considering $\Ham(y) + z_0$, we can shift the spectrum away from $\{0\}$ and this does not affect the site energies as long as we also shift the chemical potential and the contour by $z_0$. 

We now show that the Hamiltonian may be decomposed into three terms: the reference Hamiltonian and two perturbations that are small in the sense of rank and Frobenius norm, respectively. 
\begin{lemma}[Decomposition of the Hamiltonian]\label{lem:decomp_ham}Fix $y \in \emph{Adm}(\Lambda)$. For each $\delta>0$ there exists $R_\delta>0$ and operators $P_1(y), P_2(y)$ such that 
	\begin{equation}\label{eq:lem-decomp-ham} \widetilde{\Ham}(y) = \widetilde{\Ham}^\mathrm{ref} + P_1(y) + P_2(y), \end{equation}
	$\|P_1(y)\|_{\textup{\textrm{F}}} \leq \delta$ and $P_2(y)_{\ell k}^{ab} = 0$ for all $(\ell, k)\not\in B_{R_\delta}\times B_{R_\delta}$. 
\end{lemma}
\begin{proof}To simplify notation, we let $u \coloneqq y - x$. Firstly, since $u \in \W(\Lambda)$ and $y$ satisfies \asNonInter, there exists an accumulation parameter $0<\ctnoninterpen<1$ such that $|y(\ell) - y(k)| \geq \ctnoninterpen|\ell - k|$ for all $\ell, k \in \Lambda$ \cite{OrtnerShapeev12arxiv}. Moreover, since the semi-norm defined by 
\[ \|Du\|_{\ell^\infty} \coloneqq \sup_{\ell \in \Lambda} \sup_{\rho \in \Lambda - \ell} \frac{|D_\rho u(\ell)|}{|\rho|} \]is equivalent to $\|D\cdot\|_{\ell^2_\ctGamma}$ \cite{ChenNazarOrtner19}, we may choose $R > R_\mathrm{def}$ sufficiently large such that
\begin{equation}
	|D_{k-\ell}u(\ell)| \leq \ctnoninterpen|\ell - k| \quad \forall\, \ell,k\in\Lambda\setminus B_R. \label{eq:small_displacement}
\end{equation}
By applying Taylor's theorem we have: for all $\ell, k \in \Lambda\setminus B_{R}$ and atomic orbitals $1\leq a,b\leq\numorbitals$,
\begin{equation}\label{eq:ham_estimate_1}\begin{split}
	\left|\left[\widetilde{\Ham}(y) - \widetilde{\Ham}^\mathrm{ref}\right]_{\ell k}^{ab}\right| = \left|\big[{\Ham}(y) - {\Ham}(x)\big]_{\ell k}^{ab}\right| &= \left|\nabla h^{ab}_{\ell k}(\xi) \cdot \left[ u(k) - u(\ell)\right] \right|\\
	&\leq \ctTBprefactor{1} e^{-\ctTBexponent{1} |\xi|} |D_{k-\ell}u(\ell)|
\end{split}\end{equation}
where $\xi = (1-\theta)\left(y(\ell) - y(k)\right) +  \theta\left( \ell - k \right)$ for some $\theta = \theta(a,b,\ell,k) \in [0,1]$. Now, by \cref{eq:small_displacement}, we necessarily have that $|\xi| \geq \tfrac{\sqrt{3}}{2}\ctnoninterpen|\ell - k|$. In particular, by \cref{eq:ham_estimate_1}, we obtain
\begin{align}\label{eq:ham_perturb_1} \left|\left[\widetilde{\Ham}(y) - \widetilde{\Ham}^\mathrm{ref}\right]_{\ell k}^{ab}\right| \leq \ctTBprefactor{1}\,e^{-\frac{\sqrt{3}}{2}\ctTBexponent{1}\ctnoninterpen|\ell-k|}\left|D_{k-\ell}u(\ell)\right|\end{align}
for all $\ell, k \in \Lambda\setminus B_{R}$. 
	
The off-diagonal Hamiltonian entries decay exponentially and so we obtain: for $R^\prime > 0$,
	\begin{align}\label{eq:ham_perturb_2}\begin{split}
	\sum_{\ell \in (\Lambda \cup \Lambda^\mathrm{ref})\cap B_R} \sum_{\above{k \in \Lambda\setminus B_R}{|\ell - k| > R^\prime}}&\left|\left[\widetilde{\Ham}(y) - \widetilde{\Ham}^\mathrm{ref}\right]_{\ell k}\right|^2 \leq C\sum_{\ell \in (\Lambda \cup \Lambda^\mathrm{ref})\cap B_R} \sum_{\above{k \in \Lambda\setminus B_R}{|\ell - k| > R^\prime}} e^{-2\ctTBexponent{0} \ctnoninterpen|\ell - k|}\\
	&\leq C_R \int_{\above{\bm{r} \in \mathbb R^d}{|\bm{r}|>R^\prime}} e^{-2\ctTBexponent{0} \ctnoninterpen|\bm{r}|}\textrm{d}\bm{r} \leq C_R \,p(R^\prime) e^{-2\ctTBexponent{0} \ctnoninterpen R^\prime}  
	\end{split}\end{align}
where $p(R^\prime)$ is a polynomial (of degree $d-1$) in $R^\prime$. We let $P_1(y)$ be the operator (depending on $R$ and $R^\prime$) defined by 
	\begin{equation}
	P_1(y)_{\ell k}^{ab} = \begin{cases}
	\left[\Ham(y) - \Ham^\mathrm{ref}\right]_{\ell k}^{ab} &\text{if } \ell, k \in \Lambda \setminus B_{R} \\
	\left[\widetilde{\Ham}(y) - \widetilde{\Ham}^\mathrm{ref}\right]_{\ell k}^{ab} &\text{if }|\ell - k| > R^\prime \text{ and } \left(\ell \in B_R, k \not\in B_R \text{ or \textit{vice versa}}\right) \\
	\quad0 &\text{otherwise}
	\end{cases}
	\end{equation}
	for $\ell, k \in \Lambda \cup \Lambda^\mathrm{ref}$. We then define $P_2(y)$ to be the finite rank operator such that (\ref{eq:lem-decomp-ham}) is satisfied. $P_2(y)$ is local in the sense that $P_2(y)_{\ell k}^{ab} = 0$ if $(\ell,k) \not\in B_{R^{\prime\prime}}\times B_{R^{\prime\prime}}$ for some $R^{\prime\prime}$ (depending on $R$ and $R^\prime$). To conclude, we simply use (\ref{eq:ham_perturb_1}) and (\ref{eq:ham_perturb_2}) to bound $\|P_1(y)\|_\textrm{F}$:
	\begin{equation*}\begin{split}
	\|P_1(y)\|_{\textrm{F}}^2 &= \sum_{1\leq a,b\leq N_\textrm{b}}\Bigg(\sum_{\ell, k \in \Lambda\setminus B_{R}} \left| P_1(y)_{\ell k}^{ab}\right|^2 + 2\sum_{\ell \in (\Lambda\cup \Lambda^\mathrm{ref})\cap B_{R}} \sum_{\above{k \in \Lambda\setminus B_R}{|\ell - k| > R^\prime}} \left| P_1(y)_{\ell k}^{ab}\right|^2\Bigg) \\
	&\leq C \|Du\|^2_{\ell^2_\ctGamma(\Lambda\setminus B_{R})} + C_R\, p(R^\prime) e^{-2\ctTBexponent{0} \ctnoninterpen R^\prime}.
	\end{split}\end{equation*}
	Since this expression can be made arbitrarily small by choosing $R$ and then $R^\prime$ sufficiently large, this completes the proof.
\end{proof}
We now use Lemma~\ref{lem:decomp_ham} to prove that the spectrum can be decomposed as in Lemma~\ref{lem:decomp-spec}:
\begin{proof}[Proof of Lemma~\ref{lem:decomp-spec}: Decomposition of the Spectrum]

For an operator $T$, we let $\sigma_\mathrm{disc}(T)$ denote the discrete spectrum of $T$ (that is, the set of isolated eigenvalues with finite multiplicity) and $\sigma_\mathrm{ess}(T) \coloneqq \sigma(T) \setminus \sigma_\mathrm{disc}(T)$, the essential spectrum. Morevoer, for subsets $A,B \subset \mathbb R$, the Hausdorff distance between $A$ and $B$ is denoted by $\mathrm{dist}(A,B) \coloneqq \max\{ \sup_{a \in A} \mathrm{dist}(a, B), \sup_{b \in B} \mathrm{dist}(b, A) \}$. 

By Lemma~\ref{lem:decomp_ham}, we can apply Weyl's theorem \cite{Kato95} to conclude that the finite rank perturbation $P_2(y)$ does not affect the essential spectrum of $\Ham^\mathrm{ref}$. That is, $\sigma_\mathrm{ess}(\Ham^\mathrm{ref}) = \sigma_\mathrm{ess}(\Ham^\mathrm{ref} + P_2(y))$. Further, for each fixed $\eta > 0$, there are finitely many points in $\sigma_\mathrm{disc}(\Ham^\mathrm{ref} + P_2(y)) \setminus B_\eta(\sigma_\mathrm{ess}(\Ham^\mathrm{ref}))$. Indeed, since the spectrum is bounded, if $\lambda_n \in \sigma_\mathrm{disc}(\Ham^\mathrm{ref} + P_2(y))$ is a sequence of distinct eigenvalues, then (along a subsequence) we have $\lambda_n \to \lambda$ as $n \to \infty$ for some $\lambda \in \mathbb R$. Since the spectrum is closed and $\lambda$ is not an isolated point, $\lambda \in \sigma_\mathrm{ess}(\Ham^\mathrm{ref})$. This means only finitely many of $\lambda_n$ are contained in $\sigma_\mathrm{disc}(\Ham^\mathrm{ref} + P_2(y)) \setminus B_\eta(\sigma_\mathrm{ess}(\Ham^\mathrm{ref}))$.

Since $\|P_1(y)\|_\mathrm{F}\leq \delta$, $P_1(y)$ only perturbs $\sigma(\Ham^\mathrm{ref}+P_2(y))$ by $\delta$ in the following sense \cite{Kato95}:
\[
\mathrm{dist}(\sigma(\Ham(y)), \sigma(\Ham^\mathrm{ref} + P_2(y))) \leq \delta.
\]
This implies that $\sigma_\mathrm{disc}(\Ham(y)) \setminus B_{\delta + \eta}(\sigma_\mathrm{ess}(\Ham^\mathrm{ref}))$ is finite. Moreover, if $\lambda \in \sigma_\mathrm{ess}(\Ham(y))$ then we necessarily have that $\lambda \in B_{\delta + \eta}(\sigma_\mathrm{ess}(\Ham^\mathrm{ref}))$.

Since, $\Ham^\mathrm{ref} = \Ham(y) - P_1(y) - P_2(y)$ where $- P_1(y)$ and $- P_2(y)$ are small in the sense of Frobenius norm and rank, respectively, we may reverse the whole argument to conclude that $\sigma_\mathrm{ess}(\Ham^\mathrm{ref}) \subset B_{\delta + \eta}(\sigma_\mathrm{ess}(\Ham(y)))$. That is, $\mathrm{dist}(\sigma_\mathrm{ess}(\Ham(y)), \sigma_\mathrm{ess}(\Ham^\mathrm{ref}))\leq \delta + \eta$.

We have actually proved a stronger statement than that of Lemma~\ref{lem:decomp-spec}. In particular, we have shown that for all $\delta> 0$ there exists an $R_\delta$ such that $\mathrm{dist}(\sigma_\mathrm{ess}(\Ham(y)), \sigma_\mathrm{ess}(\Ham^\mathrm{ref}))\leq \delta$ and $\#\left( \sigma(\Ham(y)) \setminus B_\delta( \sigma_\mathrm{ess}(\Ham^\mathrm{ref})) \right) \leq  R_\delta$.
\end{proof}

\subsection{Proof of Theorems~\ref{thm:improved_locality_beta} and \ref{thm:improved_locality}: Improved Locality Estimates}
\label{proofs:improved_locality}

We will now prove general estimates for the resolvent operators which will be useful in the proof of the improved locality results. Firstly, we shift $\Ham(y)$ and $\Ham^\mathrm{ref}$ by the same constant multiple of the identity so that the spectrum of these operators is bounded below by a positive constant. We also fix $z \in \mathbb C$ in a bounded set such that 
\begin{equation}\label{eq:dist-dist-hom} 
\textrm{dist}\left(z, \sigma(\widetilde{\Ham}(y))\right) \geq \mathfrak{d} \quad \text{and} \quad \textrm{dist}\left(z, \sigma(\widetilde{\Ham}^\mathrm{ref})\right) \geq \mathfrak{d}^\mathrm{ref} 
\end{equation}
for positive constants $\mathfrak{d}, \mathfrak{d}^\mathrm{ref}$. In the following, we use the notation for the finite rank perturbation, $P_2(y)$, from Lemma~\ref{lem:decomp_ham} and the Combes-Thomas exponents $\ctCT(\mathfrak{d})$ and $\ctCT(\mathfrak{d}^\mathrm{ref})$ from Lemma~\ref{lem:Thomas-Combes}.

The operator $P_2(y)$ is of finite rank and so there exists $U = (w_1|\dots|w_k)$ such that $\{w_i\}_i$ is a basis of the column space of $P_2(y)$. We let $v_j$ be the vector of coordinates of the $j^\text{th}$ column of $P_2(y)$ with respect to the basis $\{w_i\}_i$ and set $V$ to be matrix with columns $v_j$. That is, $P_2(y)= UV$. Moreover, since $P_2(y)$ is local, we can choose the columns of $U$ to be such that $U_{\ell j} = 0$ for all $\ell \not\in \Lambda \cap B_R, j \in \{1,\dots,k\}$ and some $R>0$. This implies that $V_{j \ell} = 0$ for all $j \in \{1,\dots,k\}, \ell \not \in \Lambda\cap B_R$. Applying the Woodbury identity \cite{Hager1989} with $(\widetilde{\Ham}^\mathrm{ref} + P_2(y) - z)^{-1}$ and $\mathscr R^\mathrm{ref}_z \coloneqq (\widetilde{\Ham}^\mathrm{ref} - z)^{-1}$ yields
\begin{equation}\label{eq:woodbury} \left(\widetilde{\Ham}^\mathrm{ref} + P_2(y) - z\right)^{-1} = \mathscr R^\mathrm{ref}_z - \mathscr R^\mathrm{ref}_z U (I_k + V\mathscr R^\mathrm{ref}_z U)^{-1}V\mathscr R^\mathrm{ref}_z.  \end{equation}
Since $U$ and $V$ are both zero outside the finite square submatrices corresponding to the atom sites contained in ${\Lambda \cap B_R}$, we know that $U (I_k + V\mathscr R^\mathrm{ref}_z U)^{-1}V$ is also zero outside the same finite submatrix. {We may therefore choose a positive constant $c_{\mathfrak{dm}}$ such that for all $\ell, k \in \Lambda \cap B_R$,
\begin{align*}
\begin{split}
&\left|\left[ U(I + V\mathscr R^\mathrm{ref}_z U)^{-1}V \right]_{\ell k}^{ab}\right| \leq c_{\mathfrak{dm}}.
\end{split}
\end{align*}
Here, we use the fact that $z$ is contained in a bounded set for which (\ref{eq:dist-dist-hom}) is satisfied to conclude that $c_{\mathfrak{dm}}$ is independent of $z$.} 

Now we may bound the additional contribution in \cref{eq:woodbury}: 
\begin{align}\label{eq:woodbury-extra}\begin{split}
\left|\left[\mathscr R^\mathrm{ref}_z U (I_k + V\mathscr R^\mathrm{ref}_z U)^{-1}V\mathscr R^\mathrm{ref}_z\right]_{\ell k}^{ab}\right| &= \left|\sum_{\ell_1, \ell_2 \in \Lambda \cap B_R}\left[\mathscr R^\mathrm{ref}_z\right]_{\ell\ell_1}\left[ U (I_k + V\mathscr R^\mathrm{ref}_z U)^{-1}V\right]_{\ell_1\ell_2}\left[\mathscr R^\mathrm{ref}_z\right]_{\ell_2 k}\right|\\
&\leq C c_\mathfrak{dm}\big(\mathfrak{d}^\mathrm{ref}\big)^{-2} \sum_{\ell_1,\ell_2 \in \Lambda \cap B_R} e^{-\ctCT(\mathfrak{d}^\mathrm{ref}) \left(r_{\ell\ell_1} + r_{\ell_2 k}\right)}\\
&\leq C c_\mathfrak{dm}\big(\mathfrak{d}^\mathrm{ref}\big)^{-2} e^{-\ctCT(\mathfrak{d}^\mathrm{ref}) \left(|y(\ell)| + |y(k)|\right)}.
\end{split}\end{align}
Combining (\ref{eq:woodbury-extra}) and (\ref{eq:woodbury}) results in the following improved Combes-Thomas type estimate:
\begin{equation}\label{eq:improved_Combes-Thomas}
\begin{split} 
&\quad\left|\left(\widetilde{\Ham}^\mathrm{ref} + P_2(y) - z\right)^{-1}_{\ell k,ab}\right| \leq C_{\ell k} e^{-\ctCT(\mathfrak{d}^\mathrm{ref})r_{\ell k}} \quad \text{where, }\\
&C_{\ell k} = C \left\{\big(\mathfrak{d}^\mathrm{ref}\big)^{-1} + c_\mathfrak{dm} \big(\mathfrak{d}^\mathrm{ref}\big)^{-2}e^{-\ctCT(\mathfrak{d}^\mathrm{ref}) \left(|y(\ell)| + |y(k)| - r_{\ell k}\right)}\right\}.\end{split}
\end{equation}
Since, adding $P_1(y)$ only perturbs the spectrum by $\delta$ as in Lemma~\ref{lem:decomp-spec}, the same estimates hold with exponent $\ctCT(\mathfrak{d}^\mathrm{ref} -\delta)$ when $(\widetilde{\Ham}^\mathrm{ref} + P_2(y) - z)^{-1}$ is replaced by $\mathscr R_z(y)$.

Finally, we show that the estimate (\ref{eq:improved_Combes-Thomas}) implies improved estimates for derivatives of the resolvent. Moreover, we will show that the pre-factor decays away from the defect. To simplify notation further, we let 
$\mathfrak{d}^\mathrm{ref}_j \coloneqq \min\{\ctTBexponent{1},\dots,\ctTBexponent{j},\ctCT(\mathfrak{d}^\mathrm{ref})\}$
for each $j$. In place of (\ref{eq:first_derivative_resolvent}), we now have
\begin{align}\label{eq:first_derivative_resolvent_improved}
\begin{split}
\left|\frac{\partial \left[\mathscr R_z(y)\right]^{aa}_{\ell\ell}}{\partial [y(m)]_i}\right| 
&\leq C \left(\sum_{\ell_1\in \Lambda} C_{\ell\ell_1} e^{-\mathfrak{d}_1^\mathrm{ref} \left(r_{\ell \ell_1} + r_{\ell_1m}\right)} \right)^2\\
&\leq C \left\{\big(\mathfrak{d}^\mathrm{ref}\big)^{-2} \big(\mathfrak{d}_1^\mathrm{ref}\big)^{-2d} e^{-\mathfrak{d}_1^\mathrm{ref} r_{\ell m}} + c_{\mathfrak{dm}}^2\big( \mathfrak{d}^\mathrm{ref}\big)^{-4}\left(\sum_{\ell_1\in \Lambda} e^{-\mathfrak{d}_1^\mathrm{ref} \left(|y(\ell)| + |y(\ell_1)| + r_{\ell_1m}\right)} \right)^2\right\} \\
&\leq C\big(\mathfrak{d}_1^\mathrm{ref}\big)^{-2d}\left(\big( \mathfrak{d}^\mathrm{ref}\big)^{-2} + c_{\mathfrak{dm}}^2\big( \mathfrak{d}^\mathrm{ref}\big)^{-4} e^{-\mathfrak{d}_1^\mathrm{ref} \left(|y(\ell)| + |y(m)| - r_{\ell m}\right)} \right) e^{-\mathfrak{d}_1^\mathrm{ref} r_{\ell m}} \\
&\eqqcolon C(\ell, m)e^{-\mathfrak{d}_1^\mathrm{ref} r_{\ell m}} .\end{split}
\end{align}
The pre-factor, $C(\ell,m)$, in (\ref{eq:first_derivative_resolvent_improved}) converges to the corresponding pre-factor for the reference resolvent, i.e. to %
$C\big(\mathfrak{d}_1^\mathrm{ref}\big)^{-2d} \big(\mathfrak{d}^\mathrm{ref}\big)^{-2}$,
as $|y(\ell)| + |y(m)| - r_{\ell m} \to \infty$.

We will now do the same calculation for the second order derivatives of the resolvent. In this case, (\ref{eq:second_derivative_resolvent_estimate}) now takes the form
\begin{align}\label{eq:second_derivatives_resolvent_improved}
\begin{split}
&\big|\big[\mathscr R_z \Ham_{,m_1} \mathscr R_z \Ham_{,m_2} \mathscr R_z\big]_{\ell\ell}^{aa}\big|\\ 
&\leq \sum_{\ell_1,\ell_2,\ell_3,\ell_4 \in \Lambda} C_{\ell\ell_1}C_{\ell_2\ell_3}C_{\ell_4\ell}e^{-\ctCT(\mathfrak{d}^\mathrm{ref}) \left( r_{\ell\ell_1} + r_{\ell_2\ell_3} + r_{\ell_4\ell}\right)} e^{-\ctTBexponent{1} \left( r_{\ell_1m_1} + r_{m_1\ell_2} + r_{\ell_3m_2} + r_{m_2\ell_4}\right)} \\
&= C \sum_{\ell_1,\ell_2,\ell_3,\ell_4\in\Lambda}\bigg\{\big(\mathfrak{d}^\mathrm{ref}\big)^{-3}  e^{-\ctCT(\mathfrak{d}^\mathrm{ref}) \left( r_{\ell\ell_1} + r_{\ell_2\ell_3} + r_{\ell_4\ell}\right)}  \\
&\,+ c_\mathfrak{dm}\big(\mathfrak{d}^\mathrm{ref}\big)^{-4}  \Big( e^{-\ctCT(\mathfrak{d}^\mathrm{ref}) \left( |y(\ell)| + |y(\ell_1)| + r_{\ell_2\ell_3} + r_{\ell_4\ell}\right)} + e^{-\ctCT(\mathfrak{d}^\mathrm{ref}) \left( r_{\ell\ell_1} + |y(\ell_2)| + |y(\ell_3)| + r_{\ell_4\ell}\right)} \\
&\qquad\qquad\qquad\qquad\qquad+ e^{-\ctCT(\mathfrak{d}^\mathrm{ref}) \left( r_{\ell\ell_1} + r_{\ell_2\ell_3} + |y(\ell_4)| + |y(\ell)|\right)} \Big) \\
&\,+ c_\mathfrak{dm}^2\big(\mathfrak{d}^\mathrm{ref}\big)^{-5}\Big( e^{-\ctCT(\mathfrak{d}^\mathrm{ref}) \left(r_{\ell\ell_1} + |y(\ell_2)| + |y(\ell_3)| + |y(\ell_4)| + |y(\ell)| \right)} + e^{-\ctCT(\mathfrak{d}^\mathrm{ref}) \left(|y(\ell)| + |y(\ell_1)| + r_{\ell_2\ell_3} + |y(\ell_4)| + |y(\ell)|\right)} \\
&\qquad\qquad\qquad\qquad\qquad+e^{-\ctCT(\mathfrak{d}^\mathrm{ref}) \left( |y(\ell)| + |y(\ell_1)| + |y(\ell_2)| + |y(\ell_3)| + r_{\ell_4\ell}\right)} \Big)\\
&\,+ c_\mathfrak{dm}^3\big(\mathfrak{d}^\mathrm{ref}\big)^{-6}\left( e^{-\ctCT(\mathfrak{d}^\mathrm{ref}) \left( |y(\ell)| + |y(\ell_1)| + |y(\ell_2)| + |y(\ell_3)| + |y(\ell_4)| + |y(\ell)|\right)} \right)\bigg\}e^{-\ctTBexponent{1} \left( r_{\ell_1m_1} + r_{m_1\ell_2} + r_{\ell_3m_2} + r_{m_2\ell_4}\right)} \\
&\leq C\big(\mathfrak{d}_1^\mathrm{ref}\big)^{-4d}\bigg\{\big(\mathfrak{d}^\mathrm{ref}\big)^{-3} \\
&\, + c_\mathfrak{dm}\big(\mathfrak{d}^\mathrm{ref}\big)^{-4}\Big(e^{-\frac{1}{2}\mathfrak{d}_1^\mathrm{ref}\left(|y(\ell)| + |y(m_1)| - r_{\ell m_1}\right)} + e^{-\frac{1}{2}\mathfrak{d}_1^\mathrm{ref}\left(|y(m_1)| + |y(m_2)|\right)} + e^{-\frac{1}{2}\mathfrak{d}_1^\mathrm{ref}\left(|y(\ell)| + |y(m_2)| - r_{\ell m_2}\right)} \Big) \\
&\,+ c_\mathfrak{dm}^2\big(\mathfrak{d}^\mathrm{ref}\big)^{-5}\Big( e^{-\mathfrak{d}_1^\mathrm{ref}\left(|y(\ell)| + |y(m_2)| - r_{\ell m_2}\right)} +  e^{-\frac{1}{2}\mathfrak{d}_1^\mathrm{ref}\left(2|y(\ell)| + |y(m_1)| +|y(m_2)| - r_{\ell m_1} -r_{\ell m_2}\right)} \\
&\qquad\qquad\qquad\qquad\qquad+ e^{-\mathfrak{d}_1^\mathrm{ref}\left(|y(\ell)| + |y(m_1)| - r_{\ell m_1}\right)} \Big) \\
&\,+ c_\mathfrak{dm}^3 \big(\mathfrak{d}^\mathrm{ref}\big)^{-6}e^{-\mathfrak{d}_1^\mathrm{ref}\left(2|y(\ell)| + |y(m_1)| +|y(m_2)| - r_{\ell m_1} -r_{\ell m_2}\right)} \bigg\}e^{-\frac{1}{2}\mathfrak{d}_1^\mathrm{ref} \left(r_{\ell m_1} + r_{\ell m_2}\right) }\\&\eqqcolon C_{1}(\ell, m_1, m_2)e^{-\frac{1}{2}\mathfrak{d}_1^\mathrm{ref} \left(r_{\ell m_1} + r_{\ell m_2}\right) }.
\end{split}
\end{align}
Again, the pre-factor converges to the reference pre-factor, i.e. to  %
$C\big(\mathfrak{d}_1^\mathrm{ref}\big)^{-4d}\big(\mathfrak{d}^\mathrm{ref}\big)^{-3}$, 
exponentially as $|y(\ell)| + |y(m_1)| - r_{\ell m_1}$ and $|y(\ell)| + |y(m_2)| - r_{\ell m_2} \to \infty$. Similarly,
\begin{align}\label{eq:second_derivatives_resolvent_improved_2}
\begin{split}
&\big|\big[\mathscr R_z \Ham_{,m_1m_2} \mathscr R_z\big]_{\ell\ell}^{aa}\big| \leq C \sum_{\ell_1,\ell_2 \in \Lambda} c_{\ell\ell_1}c_{\ell_2\ell} e^{-\ctCT(\mathfrak{d}^\mathrm{ref}) \left( r_{\ell\ell_1} + r_{\ell_2\ell}\right)} e^{-\ctTBexponent{2} \left( r_{\ell_1m_1} + r_{\ell_1m_2}  + r_{\ell_2m_1}  r_{\ell_2m_2} \right)} \\
&= C \sum_{\ell_1,\ell_2\in\Lambda} \bigg\{\big(\mathfrak{d}^\mathrm{ref}\big)^{-2}e^{-\ctCT(\mathfrak{d}^\mathrm{ref}) \left( r_{\ell\ell_1} + r_{\ell_2\ell}\right)} \\
&\qquad+ c_\mathfrak{dm}\big(\mathfrak{d}^\mathrm{ref}\big)^{-3}\left( e^{-\ctCT(\mathfrak{d}^\mathrm{ref}) \left(|y(\ell)| + |y(\ell_1)| +r_{\ell\ell_2}\right)} + e^{-\ctCT(\mathfrak{d}^\mathrm{ref}) \left(r_{\ell\ell_1} + |y(\ell_2)| + |y(\ell)|\right)}\right) \\&\qquad+  c_\mathfrak{dm}^2\big(\mathfrak{d}^\mathrm{ref}\big)^{-4}e^{-\ctCT(\mathfrak{d}^\mathrm{ref}) \left(|y(\ell)| + |y(\ell_1)| + |y(\ell_2)| + |y(\ell)| \right)}  \bigg\} e^{-\ctTBexponent{2} \left( r_{\ell_1m_1} + r_{\ell_1m_2}  + r_{\ell_2m_1}  r_{\ell_2m_2} \right)} \\
&\leq C \big( \mathfrak{d}_2^\mathrm{ref} \big)^{-2d}\big(\mathfrak{d}^\mathrm{ref}\big)^{-2}\left[ 1 + c_\mathfrak{dm}\big(\mathfrak{d}^\mathrm{ref}\big)^{-2}e^{-\frac{1}{4} \mathfrak{d}_2^\mathrm{ref}\left(2|y(\ell)| + |y(m_1)| + |y(m_2)| - r_{\ell m_1} - r_{\ell m_2}\right) } \right]^2 e^{-\frac{1}{2}\mathfrak{d}_2^\mathrm{ref} (r_{\ell m_1} + r_{\ell m_2}) }\\&\eqqcolon C_2(\ell, m_1, m_2)e^{-\frac{1}{2}\mathfrak{d}_1^\mathrm{ref} \left(r_{\ell m_1} + r_{\ell m_2}\right) }.
\end{split}
\end{align}
Therefore, by using (\ref{eq:second_derivatives_resolvent}), we have
\begin{align}\label{eq:second_derivatives_resolvent_improved_3}
\begin{split}
\left| \frac{\partial^2 \left[\mathscr R_z(y)\right]^{aa}_{\ell\ell}}{\partial y(m_1) \partial y(m_2)} \right| \leq \max\left\{ C_1(\ell,m_1,m_2), C_1(\ell,m_1,m_2) \right\} e^{-\frac{1}{2}\mathfrak{d}_2^\mathrm{ref}\left(r_{\ell m_1} + r_{\ell m_2}\right) }
\end{split}
\end{align}
where the pre-factor converges to the pre-factor arising if $\mathscr R_z(y)$ is replaced with $\mathscr R_z^\mathrm{ref}$, i.e. it converges to 
$C \big(\mathfrak{d}^\mathrm{ref}\big)^{-3} \max\left\{ \big(\mathfrak{d}_1^\mathrm{ref}\big)^{-4d}, \mathfrak{d}^\mathrm{ref} \big(\mathfrak{d}^\mathrm{ref}_2\big)^{-2d} \right\}$, 
as we send $\ell$, $m_1$ and $m_2$ away from the defect core together. Again, we omit the arguments for $j>2$.

\begin{proof}[Proof of Theorem~\ref{thm:improved_locality}: Improved Zero Temperature Locality]
We directly apply (\ref{eq:first_derivative_resolvent_improved}) and (\ref{eq:second_derivatives_resolvent_improved_3}) with $\mathfrak{d}^\mathrm{ref} = \tfrac{1}{2}\ctgapHom$ and $\mathfrak{d} = \tfrac{1}{2}\ctgap$. We again use the fact that $\left|2(z - \mu)\right|$ is uniformly bounded along the contour $\mathscr C_\infty$.
\end{proof}

If $\mu \not\in\sigma(\Ham(y))$, similar arguments can be made for the finite temperature case. However, if $\mu \in \sigma(\Ham(y))$, another contribution to the site energy must be considered:
\begin{proof}[Proof of Theorem~\ref{thm:improved_locality_beta}: Improved Finite Temperature Locality]In the case that $\mu \not\in\sigma(\Ham(y))$, we can directly apply (\ref{eq:first_derivative_resolvent_improved}) and (\ref{eq:second_derivatives_resolvent_improved_3}) with $\mathfrak{d}^\mathrm{ref} = \ctgapfinitetempHom$ and $\mathfrak{d} = \ctgapfinitetemp(y)$. Here we again use the fact that the analytic continuation of $\mathfrak{g}^\beta(z;\mu)$ is uniformly bounded along $\mathscr C_\beta$.
	
In the case that $\mu \in \sigma(\Ham(y))$, we may split $\mathscr C_\beta$ into three simple closed contours $\mathscr C^-$, $\mathscr C^+$ and $\mathscr C_0$ such that $\mathscr C^-$ and $\mathscr C^+$ are contained in $\mathbb C\setminus (\mu + i\mathbb R)$ and encircle $\sigma(\Ham(y))\cap [\underline{\sigma},\mu)$ and $\sigma(\Ham(y))\cap(\mu,\overline{\sigma}]$, respectively, and $\mathscr C_0$ encircles $\{\mu\}$ and avoids the rest of the spectrum. Now the finite temperature site energy is of the form:
\begin{align}\label{eq:finitetempimproved}
\begin{split}
 G^\beta_{\ell}(y) &= -\frac{1}{2\pi i} \sum_a\oint_{\mathscr{C}^-} \mathfrak{g}^\beta(z;\mu) [\mathscr R_z(y)]_{\ell\ell}^{aa}\mathrm{d}z - \frac{1}{2\pi i} \sum_a\oint_{\mathscr{C}^+} \mathfrak{g}^\beta(z;\mu) [\mathscr R_z(y)]_{\ell\ell}^{aa}\mathrm{d}z\\
&\qquad\qquad -\frac{1}{2\pi i} \sum_a\oint_{\mathscr{C}_0} \mathfrak{g}^\beta(z;\mu) [\mathscr R_z(y)]_{\ell\ell}^{aa}\mathrm{d}z
\end{split}
\end{align}
The first two expressions of \cref{eq:finitetempimproved} can be treated in the exact same way as in the case where $\mu \not\in \sigma(\Ham(y))$. The additional term is 
\begin{align}\label{eq:finitetempimproved-extracontribution}
\sum_a\sum_{s=1}^{m(\mu)} \mathfrak{g}^\beta(\varepsilon_s(y);\mu) [\psi_s]_{\ell a}^2 &= \frac{2}{\beta}\log\left(\frac{1}{2}\right) \sum_a\sum_{s=1}^{m(\mu)} [\psi_{s}]_{\ell a}^2
\end{align}
where $m(\mu)$ is the multiplicity of $\mu$ as an eigenvalue of $\Ham(y)$, $\{\psi_s\}$ is basis for the eigenspace of $\mu$ and $\varepsilon_s(y)$ are the eigenvalues at $\mu$ written as functions of the configuration. We wish to show that (\ref{eq:finitetempimproved-extracontribution}) has the same locality properties as the first two terms of (\ref{eq:finitetempimproved}). 

For $j=1$, we have
\begin{align*}\label{eq:derC0}
\frac{\partial}{\partial y(m)} \left( \sum_s \mathfrak{g}^\beta(\varepsilon_s(y);\mu) [\psi_s]_{\ell a}^2 \right) &= \sum_s \left(2 f_\beta(\varepsilon_s(y) - \mu) \frac{\partial  \varepsilon_s(y)}{\partial y(m)}[\psi_s]_{\ell a}^2 + \mathfrak{g}^\beta(\varepsilon_s(y);\mu) \frac{\partial [\psi_s]_{\ell a}^2}{\partial y(m)} \right) \\
&= \sum_s \left( \frac{\partial \varepsilon_s(y)}{\partial y(m)} [\psi_s]_{\ell a}^2 - \frac{2}{\beta} \log(2) \frac{\partial [\psi_s]_{\ell a}^2}{\partial y(m)} \right) \\
&=\frac{\partial}{\partial y(m)} \left( \sum_s \left(\varepsilon_s(y) - \mu - \tfrac{2}{\beta}\log(2)\right) [\psi_s]_{\ell a}^2 \right) \\&= - \frac{1}{2\pi i} \oint_{\mathscr C_0} \left(z - \mu - \tfrac{2}{\beta}\log(2)\right) \frac{\partial [\mathscr R_z(y)]_{\ell\ell}^{aa}}{\partial y(m)} \textrm{d}z.
\end{align*}
Now, because $z \mapsto z -\mu - \tfrac{2}{\beta}\log(2)$ is analytic, there is no $\beta$-dependent restriction on the contour $\mathscr C_0$. This again allows us to apply the Woodbury identity and the Combes-Thomas type estimate on the reference resolvent to obtain improved locality results. 

For $j=2$, we have
\begin{align*}
&\frac{\partial^2}{\partial y(m_1)\partial y(m_2)} \left( \sum_s \mathfrak{g}^\beta(\varepsilon_s(y);\mu) [\psi_s]_{\ell a}^2 \right) \\
&\qquad = -\frac{1}{2\pi i}\oint_{\mathscr C_0} \left(-\tfrac{1}{4}\beta(z - \mu)^2 + z - \mu - \tfrac{2}{\beta}\log(2)\right)\frac{\partial^2\left[\mathscr R_z(y)\right]_{\ell\ell}^{aa}}{\partial y(m_1)\partial y(m_2)}\textrm{d}z.
\end{align*}
Again, we see that $z\mapsto-\tfrac{1}{4}\beta(z - \mu)^2 + z - \mu - \tfrac{2}{\beta}\log(2)$ is analytic and so we may use the improved resolvent estimates on a temperature independent contour. This results in improved locality estimates with temperature independent exponents but pre-factors of the form $C\beta$.

For higher derivatives, the same arguments can be made which gives rise to $\beta$-independent exponents but pre-factors that are of the form $C \beta^{j-1}$. 
\end{proof}

\bibliography{bib}
\bibliographystyle{abbrv}

\end{document}